\documentclass[smallextended]{svjour3}
\smartqed
\usepackage{graphicx}

\usepackage{latexsym}
\usepackage{amsmath}
\usepackage{amssymb}
\usepackage{mathtools}
\usepackage{verbatim}
\usepackage{xspace}
\usepackage{cite}
\usepackage{enumerate}
\usepackage{mathrsfs}
\usepackage[misc]{ifsym}
\usepackage{placeins}

\def\mck{{\sc Edge-$k$-Coloring}\xspace}

\def\mcpath{{\sc Edge-$2$-Coloring(Path)}\xspace}
\def\mctree{{\sc Edge-$k$-Coloring(Tree)}\xspace}
\def\mckcol{{\sc Edge-$k$-Coloring($k$-Colorable)}\xspace}
\def\mcclass{{\sc Edge-$k$-Coloring(Class)}\xspace}
\newcommand{\mcarg}[1]{{\sc Edge-$k$-Coloring(#1)}\xspace}
\def\paths{\text{\sc Path}\xspace}
\def\ceilk{\lceil\sqrt{k}\rceil}

\newcommand{\colorset}{\ensuremath{\operatorname{\mathcal{C}}}\xspace}
\newcommand{\coloring}{\ensuremath{\operatorname{\mathscr{C}}}\xspace}

\newcommand{\ivalue}{\ensuremath{v_{\text{i}}}\xspace}
\newcommand{\fvalue}{\ensuremath{v_{\text{f}}}\xspace}

\newcommand{\ec}{\ensuremath{E_{+}}\xspace}
\newcommand{\ed}{\ensuremath{E_{+}^{\text{d}}}\xspace}
\newcommand{\er}{\ensuremath{E_{-}}\xspace}
\newcommand{\es}{\ensuremath{E_{+}^{\text{s}}}\xspace}

\def\dc{d_{+}}
\def\dd{d_{+}^{\text{d}}}
\def\dr{d_{-}}
\def\ds{d_{+}^{\text{s}}}

\newcommand{\Ecrit}{\ensuremath{E_{\text{crit}}}\xspace}
\newcommand{\ecrit}{\ensuremath{e_{\text{crit}}}\xspace}

\newcommand{\NF}{\ensuremath{\operatorname{\textsc{Next-Fit}}}\xspace}
\newcommand{\FF}{\ensuremath{\operatorname{\textsc{First-Fit}}}\xspace}
\newcommand{\nf}{\ensuremath{\operatorname{\textsc{NF}}}\xspace}
\newcommand{\ff}{\ensuremath{\operatorname{\textsc{FF}}}\xspace}
\newcommand{\OPT}{\ensuremath{\operatorname{\textsc{Opt}}}\xspace}
\newcommand{\RP}{\ensuremath{\operatorname{\textsc{Rand}}}\xspace}
\newcommand{\ALG}{\ensuremath{\operatorname{\textsc{A}}}\xspace}
\newcommand{\FAIR}{\ensuremath{\operatorname{\textsc{F}}}\xspace}

\newcommand{\RAND}{\ensuremath{\operatorname{\textsc{R}}}\xspace}
\newcommand{\algo}[1]{\ensuremath{\operatorname{\textsc{#1}}}\xspace}

\newcommand{\ab}[1] {\left\vert #1\right\vert}

\DeclareMathOperator{\FFm}{FF}
\DeclareMathOperator{\NFm}{NF}

\DeclareMathOperator{\E}{\mathbb{E}}

\newcommand{\eis}{\ensuremath{E^{\text{s}}}\xspace}
\newcommand{\eid}{\ensuremath{E^{\text{d}}}\xspace}
\newcommand{\eir}{\ensuremath{E^{\text{r}}}\xspace}
\newcommand{\eic}{\ensuremath{E^{\text{c}}}\xspace}

\usepackage[dvipsnames]{xcolor}
\usepackage{tikz}
\usetikzlibrary{decorations.markings}
\usetikzlibrary{shapes.geometric}

\pgfdeclarelayer{edgelayer}
\pgfdeclarelayer{nodelayer}
\pgfsetlayers{edgelayer,nodelayer,main}

\tikzstyle{none}=[inner sep=0pt]
\definecolor{hexcolor0xffffff}{rgb}{1.000,1.000,1.000}
\definecolor{hexcolor0x000000}{rgb}{0.000,0.000,0.000}
\definecolor{hexcolor0x00ff00}{rgb}{0.000,1.000,0.000}
\definecolor{hexcolor0xffff00}{rgb}{1.000,1.000,0.000}

\tikzstyle{rn}=[circle,fill=hexcolor0xffffff,draw=hexcolor0x000000,line width=0.3 pt, scale=0.3]
\tikzstyle{rn1}=[circle,fill=hexcolor0xffffff,draw=hexcolor0xffffff,line width=0.3 pt, scale=0.3]

\tikzstyle{ed}=[-,draw=hexcolor0x000000,line width=0.700]
\tikzstyle{sed}=[densely dotted,-,draw=hexcolor0x000000,line width=0.700]
\tikzstyle{ned}=[dashed,-,draw=hexcolor0x000000,line width=0.700]

\begin{document}
\journalname{Acta Informatica}

\title{Online Edge Coloring of Paths and Trees with a Fixed Number of Colors \thanks{
A preliminary version of this paper appeared in 12th Workshop on Approximation and Online Algorithms (WAOA 2014), LNCS 8952: 181-192, 2014. \\This work was partially supported by the Villum Foundation and the Danish Council for Independent Research, Natural Sciences.}
}

\author{Lene M. Favrholdt         \and
        Jesper W. Mikkelsen
}

\institute{Lene M. Favrholdt (\Letter) \at
              Department of Mathematics and Computer Science, University of Southern Denmark, \\
Campusvej 55, 5230 Odense M, Denmark\\
Tel.: +45 65 50 23 41      \\
              \email{lenem@imada.sdu.dk}  
           \and
           Jesper W. Mikkelsen \at
              Department of Mathematics and Computer Science, University of Southern Denmark\\
              Campusvej 55, 5230 Odense M, Denmark\\
              \email{jesperwm@imada.sdu.dk}
}

\maketitle

\begin{abstract}
We study a version of online edge coloring, where the goal is to
 color as many edges as possible using only a given number, $k$, of
 available colors.
All of our results are with regard to competitive analysis.
Previous attempts to identify optimal algorithms for this problem
 have failed, even for bipartite graphs.
Thus, in this paper, we analyze even more restricted graph classes,
 paths and trees.
For paths, we consider $k=2$, and for trees, we consider any $k \geq 2$.

We prove that a natural greedy algorithm called \FF is optimal among
 deterministic algorithms, on paths as well as trees.
For paths, we give a randomized algorithm, which is optimal and better
 than the best possible deterministic algorithm.
For trees, we prove that to obtain a better competitive ratio than
 \FF, the algorithm would have to be both randomized and
 unfair (i.e., reject  edges that could have been colored), and even
 such algorithms cannot be much better than \FF.
\end{abstract}

\section{Introduction}
In the classical edge coloring problem, the edges of a graph must be
 colored using as {\em few colors} as possible, under the constraint that no
 two adjacent edges receive the same color.
There is a natural dual version of the problem where a fixed number, $k$, of
 colors is given and the goal is to color as {\em many edges} as possible,
 using at most $k$ colors.
Sometimes the classical problem is called the {\em minimization} version and
 the dual problem is called the {\em maximization} version of the problem.

In this paper, we study the following online version of the maximization problem \cite{kedge}: The edges of the graph arrive one by one (in any order), each specified by its endpoints. Immediately upon receiving an edge, the algorithm must either color
 the edge with one of the $k$ colors or reject the edge. 
The decision of which of the $k$ colors to use or to reject the edge
 is irrevocable. 
We call this problem \mck. 
For any class, {\sc Class}, of graphs, we let \mcclass denote the problem
 of \mck restricted to graphs of class {\sc Class}.
For instance, \mcpath is the online problem of properly coloring as many edges
 as possible in a path using only two colors.

\paragraph{Quality measure.}
We measure the quality of an online algorithm, $\ALG$, for \mck using
the standard notion of competitive ratio \cite{CompRatio1,
  CompRatio2}. The competitive ratio compares the performance of
$\ALG$ to that of an optimal offline algorithm, $\OPT$. We denote by
$\ALG(\sigma)$ the number of edges colored by $\ALG$ when given a
sequence, $\sigma$, of edges. Similarly, $\OPT(\sigma)$ is the number
of edges in $\sigma$ colored by $\OPT$. The algorithm $\ALG$ is said
to be \emph{$C$-competitive} if there exists a constant $b$ such that
$\ALG(\sigma)\geq C\cdot \OPT (\sigma) -b$ for any input sequence
$\sigma$. The \emph{competitive ratio}, $C_{\ALG}(k)$, of $\ALG$ is
the supremum over all $C$ for which $\ALG$ is $C$-competitive. 
The competitive ratio of $\ALG$ for \mcclass is denoted by
 $C_{\ALG}^{\text{{\sc Class}}}(k)$.

Note that by this definition, $0\leq C_{\ALG}(k)\leq 1$. In particular, upper bounds on the competitive ratio are negative results and lower bounds are positive results.

If the inequality above holds even when $b=0$, we say that $\ALG$ is \emph{strictly $C$-competitive}. This gives rise to the notion of \emph{strict competitive ratio}. 
The results in this paper are strongest possible in the sense that all positive results hold for the strict competitive ratio and all negative results hold for the competitive ratio.

For randomized algorithms, a similar definition of competitive ratio
is used but $\ALG(\sigma)$ is replaced by the expected value
$\E[\ALG(\sigma)]$.

\paragraph{Notation and terminology.}
We label the $k$ colors $1,2,\ldots , k$. For $1\leq i\leq j\leq k$, define $\colorset_{i,j}=\{i,i+1 ,\ldots , j\}$. At any fixed point in the processing of the input sequence, we
denote by $\colorset_v$ the set of colors used at edges incident to the vertex $v$. A color $i\in \colorset_{1,k}$
is said to be \emph{available} at $v$ if $i\notin \colorset_v$.
Two colorings of a graph are said to be \emph{equivalent} if one can
 be obtained from the other by renaming the colors. 

If $v$ is a vertex in the input graph, we denote by $d(v)$ the number of edges incident to $v$. An \emph{isolated edge} $e=(v,u)$ is an edge such that $d(v)=d(u)=1$ at the time when $e$ is revealed.
For any $m$, we let $\langle e_1, e_2, \ldots, e_m \rangle$ denote a
 path with $m$ edges labeled such that, for $2 \leq i \leq
 m-1$, $e_i$ is adjacent to $e_{i-1}$ and $e_{i+1}$.
A \emph{star} with $m$
edges is the complete bipartite graph $K_{1,m}$.

\paragraph{Algorithms.}
An algorithm is called \emph{fair} if it never rejects an edge unless all of the $k$ colors have already been used on adjacent edges.
In \cite{kedge}, the following two fair deterministic algorithms were studied:

\FF (\ff) uses the lowest available color for each edge. It can be viewed as the
natural greedy strategy.

\NF (\nf) remembers the last used color $c_{\text{last}}$. For each edge, it
uses the first available color in the ordered sequence $\langle c_{\text{last}}+1, \ldots , k,
1, \ldots , c_{\text{last}}\rangle$. For the very first edge, it uses the color $1$.

For \mcpath, we introduce a new family of randomized algorithms:
For $\frac12 \leq p\leq 1$, $\RP_p$ is defined as follows.
Whenever an isolated edge is revealed, $\RP_p$ uses the color $1$ with
 probability $p$ and the color $2$ with probability $1-p$. 
All non-isolated edges are colored (with the only remaining color) if
 possible.
Note that $\RP_1$ is identical to \FF.

\paragraph{Previous results.}
In \cite{kedge} it is shown that any fair algorithm for \mck has a competitive ratio of at least $2\sqrt{3}-3\approx 0.46$, and at most $\frac{1}{2}$ if it is deterministic. The lower bound is tight in the sense that \NF has a competitive ratio of exactly $2\sqrt{3}-3$. The competitive ratio of \FF is at most $\frac29(\sqrt{10}-1) \approx 0.48$. It remains an open problem whether there is an algorithm with a competitive ratio better than $2\sqrt{3}-3$.
It is also shown that no algorithm (even when
allowing randomization) has a competitive ratio better than $\frac47
\approx 0.57$.

The problem \mckcol is also studied in \cite{kedge}. When the input graph is $k$-colorable, any fair algorithm is shown to have a competitive ratio of at least $\frac12$. Again, the lower bound is tight because \NF has a competitive ratio of $\frac12$. The competitive ratio of \FF is shown to be $\frac{k}{2k-1}$. An upper bound of $\frac23$ is given for deterministic algorithms in this case. 

We remark that all of the negative results mentioned above hold even if the input graph is bipartite. Thus, contrary to offline edge coloring, the online \mck problem does not appear to be significantly easier when restricted to bipartite graphs.

It is well known that for $k=1$ (i.e., for the matching problem), the greedy algorithm is an optimal deterministic algorithm with a competitive ratio of $\frac12$.

The {\em relative worst order ratio} \cite{WRdefinition,boyar2007relative} of both the maximization and minimization version of online edge coloring is studied in \cite{kedge2}. For the maximization version, it is shown that \FF and \NF are not (strictly) comparable. This is true even when the input is restricted to bipartite graphs. For the minimization version, \FF is proven better than \NF.

The {\em minimization} version of online edge coloring is studied in
\cite{Bar-Noy}. If an online algorithm never introduces a new color
unless forced to do so, it will never use more than $2\Delta -1$
different colors on graphs of maximum degree $\Delta$. It is shown in
\cite{Bar-Noy} that no (randomized) online algorithm can do better than this, even
if the input graph is restricted to being a forest.
On any graph, an optimal offline algorithm uses at most $\Delta +1$
 colors, and on trees, $\Delta$ colors suffice.
Hence, any algorithm that introduces a new color only when necessary,
 has a competitive ratio of 2, and this is optimal.

 The problem of online {\em vertex} coloring has received much attention in
the minimization version (see 
\cite{kierstead1998coloring} for a survey). 
For interval graphs, it has also been studied in the maximization version:
It follows from a result in~\cite{seatres} that
no deterministic fair algorithm can have a competitive ratio strictly
greater than $0$, even on interval graphs.
In that paper it is also shown that, on $k$-colorable interval graphs,
any fair algorithm has a competitive ratio of at least
$\frac12$. 
In~\cite{seatresTight}, it is shown that for deterministic algorithms,
this lower bound is tight, i.e., any deterministic fair algorithm has
a competitive ratio of exactly $\frac12$ on $k$-colorable interval
graphs.
Since edge coloring is equivalent to vertex coloring of line graphs,
our results and those of~\cite{kedge} and \cite{kedge2} can also
be seen as results on vertex coloring of (subclasses of) line graphs.
In particular, edge coloring a path of
$m$ edges is equivalent to vertex coloring a
path of $m$ vertices.

A study of approximation algorithms for the {\em offline maximization} version of {\em edge} coloring for multigraphs was initiated in \cite{Faprox}. This line of work has been continued in \cite{aprox1, aprox2, aprox3SIDMA, aprox6} for both simple graphs and multigraphs.

\paragraph{Our contribution.}
For \mcpath, we give a $\frac45$-competitive randomized algorithm and
prove that this is optimal. We also show that no deterministic
algorithm can be better than $\frac23$-competitive and observe that
this upper bound is tight, since \FF is $\frac23$-competitive. 
Finally, \NF turns out to be a worst possible fair algorithm with a competitive ratio of $\frac12$.

For \mctree where $k\geq 2$, we prove that \FF is
$\frac{k-1}{k}$-competitive and that no deterministic or fair
algorithm can be better than this.
Thus, an algorithm would have to be both randomized and unfair to
achieve a better competitive ratio than \FF.
However, we show that even such algorithms cannot be better than
$\frac{k}{k+1}$-competitive. 
We also show that
any fair algorithm is $\frac{2\sqrt{k}-2}{2\sqrt{k}-1}$-competitive
and that the competitive ratio of \NF is no better than this if $k$ is
a square number.
This implies that the competitive ratio of any fair algorithm goes to
 1 as $k$ goes to infinity.

\paths and {\sc Tree} are the first examples of graph classes
for which 
 an optimal deterministic algorithm for \mck has been identified.
\paths is the first graph class for which an optimal randomized
 algorithm has been identified.
It is also the first class for which it has been proven that a
 randomized algorithm can be better than a best possible deterministic
 algorithm.

We remark that all of our results for {\sc path} extend to collections of paths. Similarly, all results for {\sc tree} extend to forests. This is so because our positive results are always for the strict competitive ratio, and because our algorithms will color a single path in a collection of paths exactly as if only the edges of that path had been revealed (similarly for trees).

\section{A Charging Technique for Proving Positive Results}
\label{lowertech}
We will now describe a simple charging technique for proving lower bounds on the competitive ratio. The technique was first used for deterministic algorithms in \cite{kedge}. 
For some $C$, $0 \leq C \leq 1$, our goal is to prove that a given
 (possibly randomized) algorithm $\ALG$ is $C$-competitive. 
Assume that the edges of a graph $G=(V,E)$ have been given in some
 order, $\sigma$, and let $E_{\OPT} \subseteq E$ be the set of edges colored in some optimal
 solution. 

The \emph{initial value} $\ivalue(e)$ of an edge, $e\in E$, is
 $\ivalue(e)=\Pr [\text{$e$ is colored by $\ALG$}]$. 
For deterministic algorithms, $\ivalue(e) \in \{0,1\}$ for all $e \in E$.
Note that by
 linearity of expectation, we have $\E[\ALG(\sigma)]=\sum_{e\in
  E}\ivalue(e)$. 

The \emph{surplus} $v_{+}(e)$ of an edge, $e\in E$, (with respect to $C$) is 
\begin{equation*}
 v_{+}(e)=\begin{cases} \ivalue(e)-C, & \text{ if $e\in E_{\OPT}$} \\
 \ivalue(e), & \text{ if $e\notin E_{\OPT}$}\end{cases}
\end{equation*}
We let $E_+ \subseteq E$ and $E_- \subseteq E$ denote the sets of
 edges with positive and negative surplus, respectively.
Clearly, $E_- \subseteq E_{\OPT}$.
For deterministic algorithms, $E_-$ is exactly those edges in
 $E_{\OPT}$ that are not colored by the algorithm, and
 $E_+$ is the set of edges colored by the algorithm (assuming
 $C<1$).
The total positive surplus $\sum_{e\in E_+}v_{+}(e)$ will be redistributed among
 the edges in $E_{-}$ according to some strategy. This strategy is
 what needs to be defined when applying the technique. 

The \emph{final value} $\fvalue(e)$ of an edge $e\in E_{\OPT}$ is the
 total value of $e$ after the redistribution of surplus. 
Since only surplus value is redistributed, $\fvalue(e) \geq C$ for all $e
 \in E_{\OPT} \setminus E_-$. 
Thus, if it can be proven that $\fvalue(e) \geq C$ for all $e
 \in E_-$, then
\begin{align*}
\label{Ccomp}
\E[\ALG(\sigma)] 
& = \sum_{e\in E}\ivalue(e)\\
& = \sum_{e\in E}\fvalue(e) \\
& \geq\sum_{e\in E_{\OPT}}\fvalue(e), \text{ since } \fvalue(e) \geq 0
\text{ for all } e \in E\\
& \geq C\cdot\OPT(\sigma).
\end{align*}
Thus, it follows that $\ALG$ is (strictly) $C$-competitive.

\section{Coloring of Paths} 
In this section, we study the \mck problem when the input graph is a
path. 
This is only interesting if $k\leq 2$, since for $k \geq 3$, any fair algorithm 
colors all edges of any path. In this paper, we consider solely the case where $k=2$, but we remark that one can use similar techniques to obtain tight bounds on the competitive ratio when $k=1$. Also, the results for \paths can be extended to graphs of maximum degree $2$.

For \mcpath, our main result is a randomized algorithm with a competitive ratio of $\frac45$ and a proof that this is optimal. Before considering randomized algorithms, we give tight lower and upper bounds on the competitive ratio of deterministic algorithms.

For 2-colorable graphs, the ratios of
 Propositions~\ref{NFlemmapath} and \ref{FFpathlower} both follow
 from~\cite{kedge}.
Clearly, the positive results carry over to paths, but for $k=2$, the
 graphs used in~\cite{kedge} for the negative results are not connected.
We give simple proofs that the negative results are
 also valid when the graph is a path. 

\begin{proposition} 
\label{NFlemmapath}
For \mcpath, \NF is a worst possible fair algorithm with
\begin{displaymath}
C_{\nf}^{\paths}(2)=\frac12\,.
\end{displaymath}
\end{proposition}

\begin{proof}
The lower bound for fair algorithms follows, since 
 each rejected edge is adjacent to exactly two colored edges, and each
 colored edge is adjacent to at most two rejected edges.

For the upper bound, consider a path $\langle e_1,\ldots , e_{2m+1}\rangle$ with $2m+1$ edges.
The adversary first reveals the odd-numbered edges
 in order of increasing indices.
\NF will alternate between the two colors. 
Afterwards, the adversary reveals all the even-numbered edges. These edges must all be rejected by \NF. Thus, the competitive ratio of \NF is at most $\frac{m+1}{2m+1}$ which tends to $\frac12$ as $m$ tends to infinity.
\qed
\end{proof}

\begin{proposition}
\label{FFpathlower}
For \mcpath, \FF is an optimal deterministic algorithm with
 $$C_{\ff}^{\paths}(2) = \frac23\,.$$
\end{proposition}

\begin{proof}
Since a path is $2$-colorable, the lower bound for \FF follows from a
 result in \cite{kedge} stating that the competitive ratio of \FF is
 $\frac{k}{2k-1}$ for the \mckcol problem.
It also follows from Lemma~\ref{rplower} below, with $p=1$.

For the upper bound, let $\algo{D}$ be a deterministic algorithm and let $n\in\mathbb{N}$. The adversary first gives $n$ disjoint paths of length two. Call these the \emph{initial paths}. Let $F=\{f_{1},\ldots , f_{n_{1}}\}$ be the set of those initial paths in which both edges have been colored by $\algo{D}$ and let $U=\{u_{1},\ldots , u_{n_{2}}\}$ be the set of those initial paths in which at least one edge has been rejected. 

In each path in $F$, both colors 1 and 2 are represented.
The adversary reveals an edge connecting the edge with the color $1$ in the path $f_i$ to the edge with the color $2$ in the path $f_{i+1}$, for $1\leq i< n_1$. These connecting edges must be rejected by $\algo{D}$ so the number of colored edges in this component is at most $2n_1$. The adversary also reveals an edge connecting $u_i$ to $u_{i+1}$, for $1\leq i < n_2$. Even if all of these connecting edges can be colored, the number of colored edges in this component is at most $2n_2-1$. 

Finally, if both $F$ and $U$ are non-empty, the adversary connects the two constructed paths by a single edge which may possibly be colored. It follows that the number of colored edges can be at most $2n_1+(2n_2-1)+1=2n$. Since the total number of edges is $3n-1$, we get an upper bound on the competitive ratio of $\frac{2n}{3n-1}$ which tends to $\frac23$ as $n$ tends to infinity.
\qed\end{proof}

Knowing that no deterministic algorithm can be better than
 $\frac23$-competitive, a natural question to ask is how good a
 randomized algorithm can be. 
To this end, we analyze the family of fair, randomized
 algorithms, $\RP_p$, defined in the introduction.

\begin{lemma}
\label{rpupper}
Let $\frac12 \leq p\leq 1$. Then,
\begin{displaymath}
  C_{\RP_p}^{\text{\paths}}(2) \leq \min\left\{p^2-p+1, \frac{2}{3}(-p^2+p+1)\right\}.
\end{displaymath}
\end{lemma}

\begin{proof}
The adversary will reveal the edges of a path $P=\langle e_1,\ldots , e_m\rangle$ with $m$ edges. Consider the following two adversary strategies for doing so:
\begin{enumerate}[(i)]
\item The adversary first reveals all edges $e_i$ with $i\equiv 1\pmod{3}$, followed by all edges $e_i$ with $i\equiv 0\pmod{3}$. Finally, all the remaining edges are revealed.

\item The adversary first reveals all the odd numbered edges and thereafter all the even numbered edges.
\end{enumerate}

If the adversary uses strategy (i), it chooses $m$ such that $3$ divides $m-1$. Note that each edge $e_i$ with $i\equiv 2\pmod{3}$ has probability $p(1-p)+(1-p)p$ of being colored. It follows that 
\begin{align*}
\E[\RP_p(P)]
&=\left(\frac{1}{3}(m-1)+1\right) + \frac{1}{3}(m-1)+\frac{2}{3}(m-1)(1-p)p\\
&=\frac23(-p^2+p+1)(m-1)+1.
\end{align*}
If the adversary uses strategy (ii), it makes sure that the number, $m$, of edges in $P$ is odd. Note that each even numbered edge has probability $p^2+(1-p)^2$ of being colored. It follows that 

\begin{align*}
\E[\RP_p(P)]&=\left(\frac{1}{2}(m-1)+1\right)+\frac{1}{2}(m-1)(p^2+(1-p)^2)\\ 
&= (p^2-p+1)(m-1)+1.
\end{align*}

Thus, if $\frac23 (-p^2+p+1)\leq p^2-p+1$, the adversary uses strategy (i) and otherwise it uses strategy (ii). By choosing $m$ sufficiently large, this proves the upper bound.
\qed\end{proof}

\begin{lemma}
\label{rplower}
Let $\frac12 \leq p\leq 1$. Then,
\begin{displaymath}
  C_{\RP_p}^{\text{\paths}}(2) \geq \min\left\{p^2-p+1, \frac{2}{3}(-p^2+p+1)\right\}.
\end{displaymath}
\end{lemma}

\begin{proof}
Let $P$ be a path and assume that the edges of $P$ are given to $\RP_p$ in some order.
Consider an edge $e$ at the time of its arrival. If two edges adjacent to $e$ have already been revealed, we say that $e$ is a \emph{critical edge}. Denote by $\Ecrit$ the critical edges of $P$. Note that since $\RP_p$ is fair, it will never reject an edge which is not critical.

We let $C=\min\left\{p^2-p+1, \frac{2}{3}(-p^2+p+1)\right\}$ and apply the charging technique described in Section~\ref{lowertech}. That is, we will define a strategy for distributing the total surplus among the edges of the path such that all edges receive a final value of at least $C$. This will imply that $\RP_p$ is $C$-competitive. Note that all non-critical edges have an initial value of $1$ and, hence, a surplus of $1-C$. Thus, $E_- \subseteq \Ecrit$.

Let $e$ be a non-critical edge. Consider the largest connected component $P_e$ induced by edges from $E\setminus \Ecrit$ containing $e$. Let $e_{\text{first}}$ be the edge in $P_e$ which was revealed first. We define $l(e)$ to be the length of the shortest path in $P_e$ containing $e$ and $e_{\text{first}}$. If $e$ is revealed as an isolated edge, then $l(e)=1$. We say that $e$ is \emph{odd} if $l(e)$ is odd and that $e$ is \emph{even} if $l(e)$ is even.
The following fact is easily proven by induction on $l(e)$. 

\begin{enumerate}[Fact:]
\item[Fact:] {\em If $e$ is odd, the probability of $e$ being colored with the color $1$ is $p$.\\ If $e$ is even, the probability of $e$ being colored with the color $1$ is $1-p$.} 
\end{enumerate}

Let $\ecrit$ be a critical edge. Denote by $e_{\text{l}}$ and $e_{\text{r}}$ the two edges adjacent to $\ecrit$. These must both be non-critical and thus must be colored by $\RP_p$. The edge $\ecrit$ will be colored if and only if $e_{\text{l}}$ and $e_{\text{r}}$ are colored with the same color. Note that the random variable denoting the color received by $e_{\text{l}}$ is independent of the random variable denoting the color received by $e_{\text{r}}$. We consider two cases:

\paragraph{Case 1: $e_{\text{l}}$ and $e_{\text{r}}$ are both odd or both even.} By the fact stated above, the probability of $\ecrit$ being colored is $p^2+(1-p)^2$. It follows that 
 $$\ivalue(\ecrit)=p^2+(1-p)^2 = 2p^2-2p+1\,.$$ 
Since $e_{\text{l}}$ and $e_{\text{r}}$ are non-critical, they both have a surplus of at least $1-C$. We will transfer a value of $\frac{1}{2}(1-C)$ from each of them to the critical edge $\ecrit$. Thus, the final value of $\ecrit$ is 
\begin{align*}
\fvalue(\ecrit) 
& \geq (2p^2-2p+1)+(1-C) \\
& =2(p^2-p+1)-C\\
& \geq C, \text{ since } C \leq p^2-p+1
\end{align*}

\paragraph{Case 2: One of $e_{\text{l}}$ and $e_{\text{r}}$ is odd and the other is even.} Without loss of generality, assume that $e_{\text{l}}$ is odd and that $e_{\text{r}}$ is even. By the fact stated above, the probability of $\ecrit$ being colored is $p(1-p)+(1-p)p$.
Thus,
 $$\ivalue(\ecrit) = 2p(1-p) = 2(-p^2+p)\,.$$
 Since $e_{\text{r}}$ is even, it must be adjacent to at least one non-critical edge $e_{\text{r}}'$. We transfer a value of $\frac{1}{2}(1-C)$ from each of $e_{\text{l}}$ and $e_{\text{r}}'$ to $\ecrit$ and a value of $1-C$ from $e_{\text{r}}$ to $\ecrit$. Transferring the entire surplus of $1-C$ from $e_{\text{r}}$ to $\ecrit$ is possible, since $e_{\text{r}}'$ is non-critical and therefore $\ecrit$ is the only critical edge adjacent to $e_{\text{r}}$. 
Thus, the final value of $\ecrit$ is 
\begin{align*}
\fvalue(\ecrit)
& = 2(-p^2+p) + 2(1-C)\\
& = 2(-p^2+p+1-C) \\
& \geq C, \text{ since } C \leq \frac{2}{3}(-p^2+p+1)
\end{align*}
\qed\end{proof}

Lemmas~\ref{rpupper} and \ref{rplower} immediately imply the following theorem.

\begin{theorem}
\label{rp}
Let $\frac12 \leq p\leq 1$. Then,
\begin{displaymath}
  C_{\RP_p}^{\text{\paths}}(2)=\min\left\{p^2-p+1, \frac{2}{3}(-p^2+p+1)\right\}.
\end{displaymath}
\end{theorem}

Theorem \ref{rp} shows that, for $p=\varphi/\sqrt{5}\approx 0.7236$, $\RP_p$ has a
 competitive ratio of $\frac45$ (where $\varphi=(1+\sqrt{5})/2$ is the golden ratio). In practice, one might prefer that $p$ is, e.g., a dyadic rational (a rational of the form $a/2^b$ for $a,b\in\mathbb{N}$). It follows from Theorem~\ref{rp} that the competitive ratio of $\RP_p$ can be made arbitrarily close to $\frac45$ by choosing a dyadic rational $p$ sufficiently close to the irrational number $\varphi/\sqrt{5}$.

We will now show that $\frac45$ is the best possible competitive ratio of \emph{any} algorithm. In fact, we show that this is true even if the algorithm knows the length of the path in advance (so that only the ordering of the edges is unknown). We will use Yao's minimax principle~\cite{Yao,BE98b}. Informally, this principle allows us to prove an upper bound of $c$ on the achievable \emph{randomized} competitive ratio by exhibiting a probability distribution over permutations of the edges of a path and showing that no \emph{deterministic} algorithm can, in expectation, color more than a fraction of $c$ of the edges of the path.

\begin{theorem}
\label{yao}
If $\algo{R}$ is a (possibly randomized) algorithm for the problem \mcpath, then
 $$C_{\algo{R}}^{\paths}(2) \leq \frac{4}{5}.$$
\end{theorem}
\begin{proof}
 Let $M\in\mathbb{N}$ be a large even integer and consider a path $P$ consisting of $\frac{5}{2}M+1$ edges. We will define a probability distribution over all permutations of the edges of $P$ by describing a randomized adversary.
 
The adversary reveals the edges of $P$ as follows: First, it reveals $M+1$ isolated edges  $\{e_1,\ldots , e_{M+1}\}$. Afterwards, the adversary picks uniformly at random a set of indices $S\subseteq\{2,\ldots , M+1\}$ such that $\ab{S}=\frac{M}{2}$. For each index $i\in S$, the adversary reveals a single edge, $e$, connecting $e_{i}$ and $e_{i-1}$ (so that $\langle e_{i-1}, e, e_{i}\rangle $ becomes a subpath of $P$). 
Let $\overline{S} = \{1,2,\ldots,M+1\} \setminus S$.
For each index $i\in \overline{S}$, the adversary reveals two edges, $e$ and $e'$, connecting $e_{i}$ and $e_{i-1}$ (so that $\langle e_{i-1},e,e',e_i\rangle$ becomes a subpath of $P$). Note that the resulting path $P$ has $M+1+\frac{M}{2}+M=\frac{5}{2}M+1$ edges.

Let $\algo{D}$ be any deterministic algorithm, and let $\E[\algo{D}(P)]$ denote the expected number of edges colored by $\algo{D}$ when the edges of $P$ are revealed as described above. We will show that
 $\E[\algo{D}(P)]$ is at most \mbox{$\frac{4}{5}\OPT(P)+1$}. 
Since by Yao's principle, $C_{\algo{R}}^{\paths}(2) \leq \E[\algo{D}(P)]/\OPT(P)$ (and $\E[\algo{D}(P)]$ can be arbitrarily large), this will complete the proof. 

We first introduce some terminology to describe a coloring produced by $\algo{D}$.
For any $i$, $2 \leq i \leq M+1$, we say that $e_{i-1}$ is the \emph{previous} isolated edge of $e_{i}$. 
The set of isolated edges is partitioned into the following four sets:
\begin{enumerate}[\eis:]
\item[\eis:] Isolated edges colored with the {\em same} color as the previous isolated edge. 
\item[\eid:] Isolated edges colored {\em differently} from the previous isolated edge. 
\item[\eir:] Isolated edges that are {\em rejected}.
\item[\eic:] Isolated edges that are {\em colored} but whose previous isolated edge is {\em rejected}. 
\end{enumerate}
Clearly, $|\eic| \leq |\eir|$. 

Let $X$ be a random variable denoting the total number of edges rejected by $\algo{D}$. We will give a lower bound on $\E[X]$.
For each isolated edge $e_i$ with $2 \leq i
\leq M+1$, consider the probability of at least one of $e_i$ and the edge(s)
connecting $e_i$ to $e_{i-1}$ being rejected. For each edge in
$\eis$, the algorithm $\algo{D}$ makes a rejection with probability
$\frac{1}{2}$, since it will be forced to do so if
$i\in \overline{S}$. Conversely, for each edge in $\eid$, the algorithm $\algo{D}$
makes a rejection with probability $\frac{1}{2}$, since it is forced to
do so if $i\in{S}$. Also, for each edge in $\eir$, the algorithm
$\algo{D}$ makes a rejection with probability $1$. Combining these observations with the linearity of expectation, we get that
\begin{align*}
\E[X]&\geq \left(\frac{1}{2}\ab{\eis} + \frac{1}{2}\ab{\eid} + \ab{\eir}\right) \\
&\geq \frac{1}{2}\left(\ab{\eis} + \ab{\eid} + \ab{\eir} + \ab{\eic}\right),
      \text{ since } |\eir| \geq |\eic|\\
&\geq \frac{M}{2}\,.
\end{align*}
Finally, since \OPT can color all $\frac{5}{2}M+1$ edges of the path, we get that 
 $$\E[\algo{D}(P)]\leq 2M+1 < \frac45 \OPT (P) + 1\,.$$
Since $M$ can be arbitrarily large, this proves the theorem.
\qed\end{proof}

Theorems~\ref{rp} and \ref{yao} together give the following corollary.

\begin{corollary}
For $p=\frac{\varphi}{\sqrt{5}}$, $\RP_p$ is optimal for \mcpath with\vspace{-1mm}
\begin{displaymath}
C_{\RP_{p}}^{\text{\sc path}}(2)=\frac45.
\end{displaymath}
\end{corollary}

\section{Coloring of Trees}
We will now consider the \mck problem when the input graph is a tree. 
Our main result is a proof that \FF is optimal among deterministic
 as well as fair algorithms. 
We also show that even randomized algorithms that are not fair can
only be slightly better that \FF.
Finally, we show that, for any fixed $k\geq 4$, \FF has a 
 better competitive ratio than \NF.

First, we give a general upper bound for algorithms that are deterministic and/or fair.

\begin{theorem}
\label{detupper}
If $\ALG$ is a deterministic or fair algorithm and $k\geq 2$, then 
\begin{displaymath}
C_{\ALG}^{\text{\sc Tree}}(k)\leq\frac{k-1}{k}.
\end{displaymath}
\end{theorem}
\begin{proof}
The adversary reveals the edges of a tree in $N$ steps, for some large $N \in \mathbb{N}$. The set of edges revealed in the $i$th step constitute a star, $S_i$, with $k+1$ edges and center vertex $c_i$. If at least one edge in $S_{i-1}$ is colored, the adversary chooses $c_i = x$ for some colored edge $(c_{i-1},x)$ in $S_{i-1}$. Otherwise, it chooses $c_i = x$ for an arbitrary edge $(c_{i-1},x)$ in $S_{i-1}$. Note that the adversary is clearly able to identify a colored edge in $S_{i-1}$, if one exists: If \ALG is deterministic, this is trivially true, and if \ALG is fair,  the first $k-1$ edges of $S_{i-1}$ will be colored. 

The algorithm $\ALG$ may color $k$ edges of $S_1$. For all other values of $i$, there are two possibilities:
\begin{itemize}
\item  If $\ALG$ colors even a
single edge of $S_{i-1}$, then it can color at most $k-1$ edges of
$S_i$.
\item Even if $\ALG$ rejects all edges of $S_{i-1}$, then it can color at most $k$ edges of $S_i$. 
\end{itemize}
Let $N_0$ denote the number of stars where \ALG colors no edges.
Then, \ALG colors at most $(N_0+1)k + (N-2N_0-1)(k-1) = N(k-1) - (k-2)N_0 +1 \leq N(k-1)+1$ edges.
On the other hand, in each star, \OPT colors the $k$ edges not incident to other stars, in total $Nk$ edges. Since $N$ can be arbitrarily large, this shows that the competitive ratio of $\ALG$ is at most $\frac{k-1}{k}$.
\qed\end{proof}

Using the charging technique of Section~\ref{lowertech}, we will show that
 Theorem~\ref{detupper} is tight by proving a matching lower bound for \FF.
To this end, we introduce some terminology related to deterministic
 algorithms.

Let $\ALG$ be a deterministic algorithm for \mck, let $G=(V,E)$ be
 a graph, and suppose that $\ALG$ has been given the edges of $G$ in
 some order.
Recall that, since \ALG is deterministic, $E_+$ denotes the set of
 edges colored by \ALG, and $E_-$ denotes the set of edges
 colored by \OPT only.
We partition $E_+$ into the set, \ed, of edges colored by both \ALG
 and \OPT ({\em double colored} edges) and the set, \es, of edges 
 colored by \ALG only ({\em single colored} edges).
Thus, $E_{\OPT}=\er\cup\ed$. 
For $x\in V$, let $\ec(x)$ be the edges in $\ec$ incident to 
$x$ and let $\dc(x)=\ab{\ec(x)}$. Define $\er(x), \ed(x), \es(x), \dr(x), \dd(x)$ and $\ds(x)$ similarly.

\begin{theorem}
\label{fftree}
For $k\geq 2$, \FF is an optimal deterministic algorithm for \mctree with
$$C_{\FFm}^{\text{\sc Tree}}(k)=\frac{k-1}{k}.$$
\end{theorem}
\def\cmis{\widehat{c}_x}
\def\cmisv{\widehat{c}_v}

\begin{proof}
Fix a tree $T=(V,E)$ and assume that the edges of $E$ have been revealed to \FF in some order. For the analysis, we will view $T$ as a rooted tree by choosing an arbitrary vertex to be the root. When writing $e=(x,y)\in E$, we imply that $x$ is the parent vertex of $y$. 

Following Section~\ref{lowertech}, we set $C=\frac{k-1}{k}$. An edge in $\ed$ then has a surplus of $1-C=\frac{1}{k}$ and an edge in $\es$ has a surplus of $1$. On the other hand, an edge in $\er$ has an initial value of zero. 

We will define a strategy to distribute the total positive surplus obtained by \FF among the edges in $\er$ such that each edge gets a final value of at least $C$. For ease of presentation, the strategy will be described in a stepwise manner (see Fig.~\ref{surplusfigure} for an illustration of how the strategy works):

\begin{enumerate}[Step 1:]
\item  Consider in turn all edges $e=(v,u)\in\ec$. Let $c$ be the color assigned to $e$ by \FF and let $e'=(w,v)$ be the parent edge of $e$ (if it exists). 
  \begin{enumerate}
  \item If $e'\in \ed$ and $e'$ has been colored with a color $c'>c$,
    then $e$ transfers a value of $\frac1k$ to $w$.
  \item Any surplus remaining at $e$ is
    transferred to $v$.
  \end{enumerate}
  For each vertex $v$, let $m(v)$ denote the value transferred
   to $v$ in this step.
\item Consider in turn all vertices $v\in V$. 
  \begin{enumerate}
  \item If the vertex $v$ has a parent edge $e'\in\er$, then $v$ transfers a
    value of $\min\left\{m(v), \frac{k-1}{k}\right\}$ to $e'$. 
  \item Any value remaining at $v$ is distributed equally among the
    child edges of $v$ belonging to $\er$.
  \end{enumerate}
  For each edge $e$, let $m_v(e)$ denote the value transferred from $v$ to $e$ in this step.
\end{enumerate}

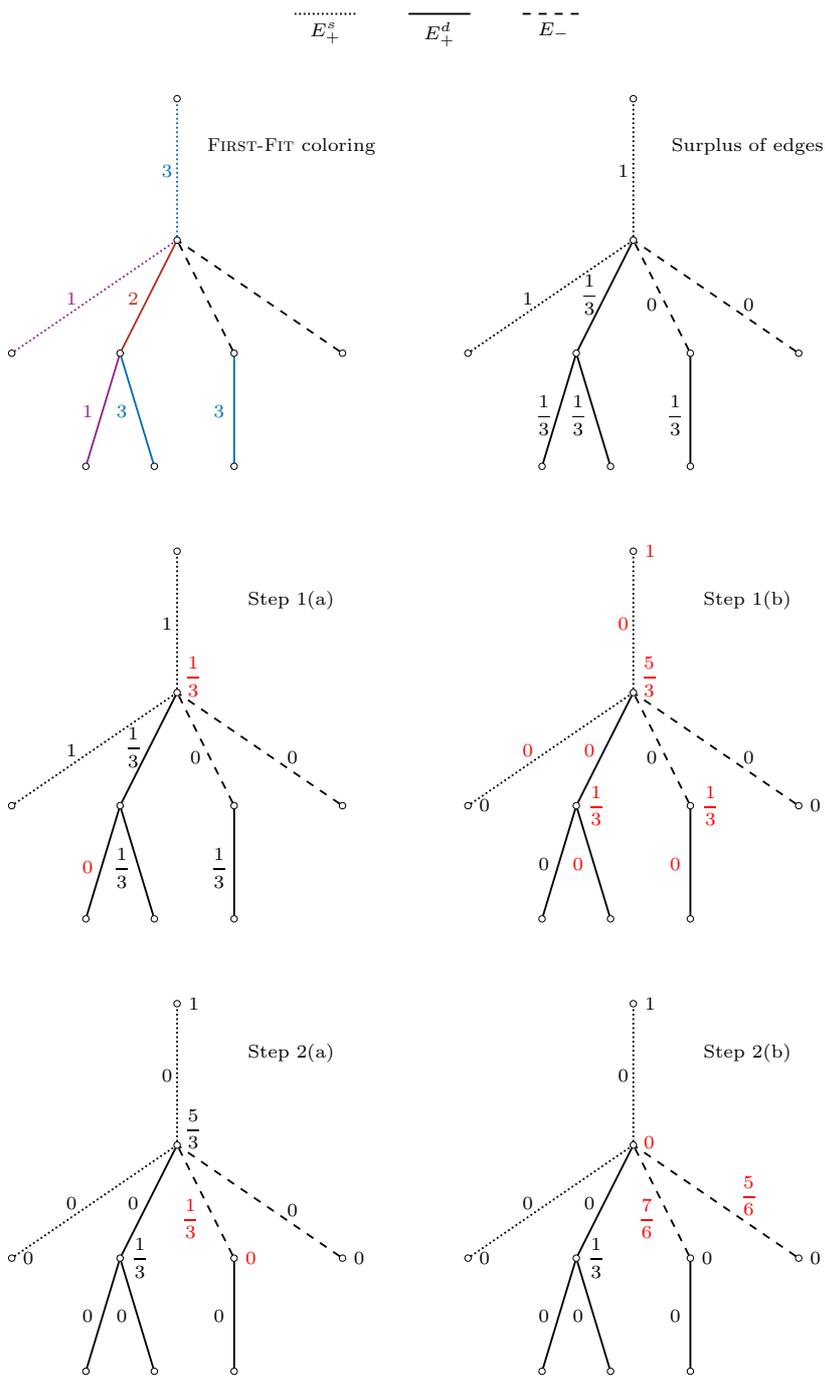
\begin{figure}
\begin{tikzpicture}[font=\scriptsize, scale=1.5]

	\begin{pgfonlayer}{nodelayer}

 \node [style=rn1, label={below:}] (le1) at (1,2) {};
 \node [style=rn1, label={below:}] (le2) at (1.6,2) {};
 \draw [style=sed] (le1) to node[label={below:$E_+^s$}]{} (le2);

 \node [style=rn1, label={below:}] (le1) at (2,2) {};
 \node [style=rn1, label={below:}] (le2) at (2.6,2) {};
 \draw [style=ed] (le1) to node[label={below:$E_+^d$}]{} (le2);

 \node [style=rn1, label={below:}] (le1) at (3,2) {};
 \node [style=rn1, label={below:}] (le2) at (3.6,2) {};
 \draw [style=ned] (le1) to node[label={below:$E_-$}]{} (le2);

 \node [style=rn1, label={below:\FF coloring}] (la1) at (1,1) {};
%\node [style=rn] (a1) at (1.3,0.5) {};
%\node [style=rn] (a1) at (2.8,0.5) {};	
	\node [style=rn,] (0) at (0, -0) {};
		\node [style=rn] (1) at (0, 1.25) {};
		\node [style=rn] (2) at (1.45, -1) {};
		\node [style=rn] (3) at (0.5, -1) {};
                \node [style=rn] (3a) at (0.5, -2) {};
		\node [style=rn] (4a) at (-0.8, -2) {};
		\node [style=rn] (4b) at (-0.2, -2) {};
                \node [style=rn] (4) at (-0.5, -1) {};
		\node [style=rn] (5) at (-1.45, -1) {};

	\end{pgfonlayer}
	\begin{pgfonlayer}{edgelayer}
		\draw [style=sed, color=NavyBlue] (0) to node[label={[xshift=-0.13cm, yshift=-0.21cm]$3$}]{}  (1);
		\draw [style=ed, color=BrickRed] (4) to node[label={[xshift=-0.2cm, yshift=-0.21cm]$2$}]{} (0);
		\draw [style=sed, color=Plum] (5) to node[label={[xshift=-0.3cm, yshift=-0.21cm, color=Plum]$1$}]{} (0);
		\draw [style=ned,  ] (3) to node[label={[xshift=-0.13cm, yshift=-0.3cm,color=red]}]{} (0);
		\draw [style=ned,  ] (2) to node[label={[xshift=0.43cm, yshift=-0.3cm,color=red]}]{} (0);

                \draw [style=ed, color=Plum] (4) to node[label={[xshift=-0.2cm, yshift=-0.21cm, color=Plum]$1$}]{} (4a);
                \draw [style=ed, color=NavyBlue] (4) to node[label={[xshift=-0.2cm, yshift=-0.21cm]$3$}]{} (4b);
                \draw [style=ed, color=NavyBlue] (3) to node[label={[xshift=-0.2cm, yshift=-0.21cm]$3$}]{} (3a);
	
	\end{pgfonlayer}

\begin{scope}[shift={(4,0)}]
	\begin{pgfonlayer}{nodelayer}

 \node [style=rn1, label={below:{Surplus of edges}}] (la1) at (1,1) {};
		\node [style=rn, ] (0) at (0, -0) {};
		\node [style=rn,] (1) at (0, 1.25) {};
		\node [style=rn,] (2) at (1.45, -1) {};
		\node [style=rn,] (3) at (0.5, -1) {};
                \node [style=rn,] (3a) at (0.5, -2) {};

		\node [style=rn] (4a) at (-0.8, -2) {};
		\node [style=rn] (4b) at (-0.2, -2) {};
                \node [style=rn,] (4) at (-0.5, -1) {};
		\node [style=rn,] (5) at (-1.45, -1) {};

	\end{pgfonlayer}
	\begin{pgfonlayer}{edgelayer}
		\draw [style=sed] (0) to node[label={[xshift=-0.13cm, yshift=-0.21cm,color=black]$1$}]{}  (1);
		\draw [style=ed] (4) to node[label={[xshift=-0.2cm, yshift=-0.35cm,color=black]$\dfrac13$}]{} (0);
		\draw [style=sed] (5) to node[label={[xshift=-0.3cm, yshift=-0.21cm,color=black]$1$}]{} (0);
		\draw [style=ned,  ] (3) to node[label={[xshift=-0.13cm, yshift=-0.3cm,color=black]$0$}]{} (0);
		\draw [style=ned,  ] (2) to node[label={[xshift=0.43cm, yshift=-0.3cm,color=black]$0$}]{} (0);

                \draw [style=ed] (4) to node[label={[xshift=-0.2cm, yshift=-0.45cm,color=black]$\dfrac13$}]{} (4a);
                \draw [style=ed] (4) to node[label={[xshift=-0.2cm, yshift=-0.45cm,color=black]$\dfrac13$}]{} (4b);
                \draw [style=ed] (3) to node[label={[xshift=-0.2cm, yshift=-0.45cm,color=black]$\dfrac13$}]{} (3a);
	
	\end{pgfonlayer}
\end{scope}

\begin{scope}[shift={(0,-4)}]
	\begin{pgfonlayer}{nodelayer}
                \node [style=rn1, label={below:Step 1(a)}] (la1) at (1,1) {};
		\node [style=rn, label={[xshift=0.21cm, yshift=-0.2cm,color=red]$\dfrac13$}] (0) at (0, -0) {};
		\node [style=rn,] (1) at (0, 1.25) {};
		\node [style=rn,] (2) at (1.45, -1) {};
		\node [style=rn,] (3) at (0.5, -1) {};
                \node [style=rn, label=right:] (3a) at (0.5, -2) {};
		\node [style=rn] (4a) at (-0.8, -2) {};
		\node [style=rn] (4b) at (-0.2, -2) {};
                \node [style=rn,] (4) at (-0.5, -1) {};
		\node [style=rn,] (5) at (-1.45, -1) {};
	\end{pgfonlayer}
	\begin{pgfonlayer}{edgelayer}
		\draw [style=sed] (0) to node[label={[xshift=-0.13cm, yshift=-0.21cm,color=black]$1$}]{}  (1);
		\draw [style=ed] (4) to node[label={[xshift=-0.2cm, yshift=-0.35cm,color=black]$\dfrac13$}]{} (0);
		\draw [style=sed] (5) to node[label={[xshift=-0.3cm, yshift=-0.21cm,color=black]$1$}]{} (0);

				\draw [style=ned,  ] (3) to node[label={[xshift=-0.13cm, yshift=-0.3cm,color=black]$0$}]{} (0);
		\draw [style=ned,  ] (2) to node[label={[xshift=0.43cm, yshift=-0.3cm,color=black]$0$}]{} (0);

                \draw [style=ed] (4) to node[label={[xshift=-0.2cm, yshift=-0.25cm,color=red]$0$}]{} (4a);
                \draw [style=ed] (4) to node[label={[xshift=-0.2cm, yshift=-0.45cm,color=black]$\dfrac13$}]{} (4b);
                \draw [style=ed] (3) to node[label={[xshift=-0.2cm, yshift=-0.45cm,color=black]$\dfrac13$}]{} (3a);
	
	\end{pgfonlayer}
\end{scope}

\begin{scope}[shift={(4,-4)}]
	\begin{pgfonlayer}{nodelayer}
           \node [style=rn1, label={below:Step 1(b)}] (la1) at (1,1) {};
	\node [style=rn, label={[xshift=0.21cm, yshift=-0.2cm,color=red]$\dfrac53$}] (0) at (0, -0) {};
		\node [style=rn, label={right:{\color{red}$1$}}] (1) at (0, 1.25) {};
		\node [style=rn, label={right:$0$}] (2) at (1.45, -1) {};
		\node [style=rn, label={right:{\color{red}$\dfrac13$}}] (3) at (0.5, -1) {};
                \node [style=rn, label=right:] (3a) at (0.5, -2) {};
		\node [style=rn] (4a) at (-0.8, -2) {};
		\node [style=rn] (4b) at (-0.2, -2) {};
                \node [style=rn, label={right:{\color{red}$\dfrac13$}}] (4) at (-0.5, -1) {};
		\node [style=rn, label={right:${\color{black}0}$}] (5) at (-1.45, -1) {};
	\end{pgfonlayer}
	\begin{pgfonlayer}{edgelayer}
		\draw [style=sed] (0) to node[label={[xshift=-0.13cm, yshift=-0.21cm,color=red]$0$}]{}  (1);
		\draw [style=ed] (4) to node[label={[xshift=-0.2cm, yshift=-0.21cm,color=red]$0$}]{} (0);
		\draw [style=sed] (5) to node[label={[xshift=-0.3cm, yshift=-0.21cm,color=red]$0$}]{} (0);

				\draw [style=ned,  ] (3) to node[label={[xshift=-0.13cm, yshift=-0.3cm,color=black]$0$}]{} (0);
		\draw [style=ned,  ] (2) to node[label={[xshift=0.43cm, yshift=-0.3cm,color=black]$0$}]{} (0);

                \draw [style=ed] (4) to node[label={[xshift=-0.2cm, yshift=-0.21cm,color=black]$0$}]{} (4a);
                \draw [style=ed] (4) to node[label={[xshift=-0.2cm, yshift=-0.21cm,color=red]$0$}]{} (4b);
                \draw [style=ed] (3) to node[label={[xshift=-0.2cm, yshift=-0.21cm,color=red]$0$}]{} (3a);
	
	\end{pgfonlayer}
\end{scope}

\begin{scope}[shift={(0,-8)}]
	\begin{pgfonlayer}{nodelayer}
           \node [style=rn1, label={below:Step 2(a)}] (la1) at (1,1) {};
           	\node [style=rn, label={[xshift=0.21cm, yshift=-0.2cm,color=black]$\dfrac53$}] (0) at (0, -0) {};
		\node [style=rn, label={right:{\color{black}$1$}}] (1) at (0, 1.25) {};
		\node [style=rn, label={right:$0$}] (2) at (1.45, -1) {};
		\node [style=rn, label={right:{\color{red}$0$}}] (3) at (0.5, -1) {};
                \node [style=rn, label=right:] (3a) at (0.5, -2) {};
		\node [style=rn] (4a) at (-0.8, -2) {};
		\node [style=rn] (4b) at (-0.2, -2) {};
                \node [style=rn, label={right:{\color{black}$\dfrac13$}}] (4) at (-0.5, -1) {};
		\node [style=rn, label={right:${\color{black}0}$}] (5) at (-1.45, -1) {};
	\end{pgfonlayer}
	\begin{pgfonlayer}{edgelayer}
		\draw [style=sed] (0) to node[label={[xshift=-0.13cm, yshift=-0.21cm,color=black]$0$}]{}  (1);
		\draw [style=ed] (4) to node[label={[xshift=-0.2cm, yshift=-0.21cm,color=black]$0$}]{} (0);
		\draw [style=sed] (5) to node[label={[xshift=-0.3cm, yshift=-0.21cm,color=black]$0$}]{} (0);

		\draw [style=ned,] (3) to node[label={[xshift=-0.2cm, yshift=-0.6cm,color=red]$\dfrac13$}]{} (0);
		\draw [style=ned,  ] (2) to node[label={[xshift=0.43cm, yshift=-0.3cm,color=black]$0$}]{} (0);

                \draw [style=ed] (4) to node[label={[xshift=-0.2cm, yshift=-0.21cm,color=black]$0$}]{} (4a);
                \draw [style=ed] (4) to node[label={[xshift=-0.2cm, yshift=-0.21cm,color=black]$0$}]{} (4b);
                \draw [style=ed] (3) to node[label={[xshift=-0.2cm, yshift=-0.21cm,color=black]$0$}]{} (3a);
	
	\end{pgfonlayer}
\end{scope}

\begin{scope}[shift={(4,-8)}]
	\begin{pgfonlayer}{nodelayer}
           \node [style=rn1, label={below:Step 2(b)}] (la1) at (1,1) {};
		\node [style=rn, label={[xshift=0.21cm, yshift=-0.2cm,color=red]$0$}] (0) at (0, -0) {};
		\node [style=rn, label={right:$1$}] (1) at (0, 1.25) {};
		\node [style=rn, label={right:$0$}] (2) at (1.45, -1) {};
		\node [style=rn, label={right:$0$}] (3) at (0.5, -1) {};
                \node [style=rn, label=right:] (3a) at (0.5, -2) {};
		\node [style=rn] (4a) at (-0.8, -2) {};
		\node [style=rn] (4b) at (-0.2, -2) {};
                \node [style=rn, label={right:$\dfrac13$}] (4) at (-0.5, -1) {};
		\node [style=rn, label={right:$0$}] (5) at (-1.45, -1) {};
	\end{pgfonlayer}
	\begin{pgfonlayer}{edgelayer}
		\draw [style=sed] (0) to node[label={[xshift=-0.13cm, yshift=-0.21cm,color=black]$0$}]{}  (1);
		\draw [style=ed] (4) to node[label={[xshift=-0.2cm, yshift=-0.21cm,color=black]$0$}]{} (0);
		\draw [style=sed] (5) to node[label={[xshift=-0.3cm, yshift=-0.21cm,color=black]$0$}]{} (0);

		\draw [style=ned,] (3) to node[label={[xshift=-0.2cm, yshift=-0.6cm,color=red]$\dfrac76$}]{} (0);
		\draw [style=ned,] (2) to node[label={[xshift=0.43cm, yshift=-0.3cm,color=red]$\dfrac56$}]{} (0);

                \draw [style=ed] (4) to node[label={[xshift=-0.2cm, yshift=-0.21cm,color=black]$0$}]{} (4a);
                \draw [style=ed] (4) to node[label={[xshift=-0.2cm, yshift=-0.21cm,color=black]$0$}]{} (4b);
                \draw [style=ed] (3) to node[label={[xshift=-0.2cm, yshift=-0.21cm,color=black]$0$}]{} (3a);
	
	\end{pgfonlayer}
\end{scope}
\end{tikzpicture}
\caption{Illustration of the steps of the strategy defined in the proof of Theorem~\ref{fftree}. In this example, the number of colors is $k=3$.}
\label{surplusfigure}
\end{figure}

The following simple but useful properties of the strategy defined
above will be used to prove the theorem.
Each of the four facts gives a lower bound on the value transferred
 from an edge $e_+ = (v,u) \in E_+$ to its parent vertex, $v$.
Let $c$ denote the color of $e_+$.
We first state the four facts and then give short proofs.

Let $e'=(w,v)$ be the parent edge of $e_+$ (if it exists). 
If $e' \in E_+$, let $c'$ denote the color of $e'$.

\begin{enumerate}[Fact 1:]
\item[Fact 1:]
{Assume that $e_+ \in \es$. \\
 If $e'\notin \ed$, $e'$ does not exist, or $c \geq c'$, then $e_+$ contributes
  a value of $1$ to $m(v)$.\\ 
 If $e'\in \ed$, then $e_+$ contributes a value of at least
  $\frac{k-1}{k}$ to $m(v)$.
 }
\item[Fact 2:]
{If $e' \not\in\ec$ or $e'$ does not exist, then $m(v)\geq
  \frac{c}{k}$. 
 } 
\item[Fact 3:]
{Assume that $e_+\in\ed$.\\
 If $e'\notin \ed$, then $e_+$ contributes a value of $\frac1k$ to
  $m(v)$.
 }
\end{enumerate}

In order to state the next fact, we need to introduce some new
terminology. For $v\in V$, let $\cmisv
=\max\left(\overline{\colorset_v}\cup\{0\}\right)$. That is, $\cmisv$ is the
largest color available at $v$ (and $\cmisv=0$ if no colors are
available). If an edge incident to $v$ is colored with a
color $c > \cmisv$, the edge is said to be a \emph{high-colored} edge (with respect to $v$). There must be exactly $k-\cmisv$ high-colored edges incident to $v$.

\begin{enumerate}[Fact 4:]
\item[Fact 4:]
{Assume that $e_+\in \ed$.\\
 If $e_+$ is high-colored with respect to $v$, then the colored child
  edges of $e_+$ contribute a total value of at least
  $\frac{k-\dc(v)}{k}$ to $m(v)$.
 }
\end{enumerate}

{\em Proof of Fact 1:}
If $e'\notin \ed$, $e'$ does not exists, or $c \geq c'$, then $e_+$ transfers a value of $1$ to $v$ in Step 1(b). If $e'\in\ed$, then $e_+$ transfers a value of at most $\frac1k$ to $w$ in Step 1(a) and hence $e_+$ transfers a value of at least $\frac{k-1}{k}$ to $v$ in Step 1(b).

{\em 
Proof of Fact 2:}
If $e_+\in \es$, this follows from Fact~1. Otherwise, note that by the definition of \FF, it must hold that $\colorset_{1,c}\subseteq \colorset_v\cup \colorset_u$. In Step 1, the edges incident to $v$ and $u$ colored with a color in $\colorset_{1,c}$ each transfer a value of at least $\frac1k$ to $v$.

{\em Proof of Fact 3:} This follows, since $e_+$ does not transfer any value to $w$ in Step 1(a).

{\em Proof of Fact 4:}
Since $e_+$ is high-colored, it follows from the definition of \FF that
 all colors in $\overline{\colorset_v}$ are represented at child edges of $u$.
Thus, $e_+$ has at least $\ab{\overline{\colorset_v}}=k-\dc(v)$ child
 edges with lower colors than the color of $e_+$.
Since $e_+ \in \ed$, each of these child edges transfers a value of $\frac1k$
 to $v$ in Step 1(a).

We will combine these facts to show that any edge $e=(x,y)\in
 \er$ gets a final value of at least $\frac{k-1}{k}$. 

If $\colorset_x=\colorset_{1,k}$, then $\cmis=0$.
Otherwise, $\cmis \in \colorset_y$, since \FF is fair.
Hence, Fact~2 implies that $m(y)\geq\frac{\cmis}{k}$. 
Thus, $e$ receives a value of at least $\min \{ \frac{k-1}{k},
 \frac{\cmis}{k} \}$ from $y$. 
In particular, we will assume that $\cmis<k-1$, since otherwise we are
 done. 
Thus, 
\begin{equation*}
m_y(e) \geq \frac{\cmis}{k}
\end{equation*}
We will now turn to proving that $m_x(e) \geq \frac{k-\cmis-1}{k}$.
This will finish the proof, since it
 means that $e$ gets a final value of
 $m_x(e) + m_y(e) \geq \frac{k-\cmis-1}{k} + \frac{\cmis}{k} = \frac{k-1}{k}$.

Let $e'=(z,x)$ be the parent edge of $x$ (if it exists). The rest of the proof is split into three cases depending on which of the sets $\ed$, $\es$, and $\er$ (if any) that contains $e'$.

\paragraph{Case 1: $e'\in \ed$.}
Recall that there are $k-\cmis$ high-colored edges incident to $x$.
Thus, $x$ has at least $k-\cmis-1$ high-colored child edges, and at
 least $k-\cmis-1-\ds(x)$ of them belong to $\ed$.
By Fact~4, $x$ receives a value of at least $\frac{k-\dc(x)}{k}$ from
 the child edges of each of these at least $k-\cmis-1-\ds(x)$ edges.
Moreover, by Fact~1, each of the $\ds(x)$ child edges of $e'$
 belonging to $\es$ contributes a value of $\frac{k-1}{k}$ to $m(x)$.
Thus, 
\begin{align*}
m(x)
& \geq (k-\cmis-1-\ds(x))\frac{k-\dc(x)}{k} + \ds(x)\frac{k-1}{k} \\
& =    (k-\cmis-1-\ds(x))\frac{k-\dc(x)}{k} +
       \ds(x) \left( \frac{k-\dc(x)}{k} + \frac{\dc(x)-1}{k} \right) \\
& =    (k-\cmis-1)\frac{k-\dc(x)}{k} + \ds(x)\frac{\dc(x)-1}{k} \\
& \geq (k-\cmis-1)\frac{k-\dc(x)}{k} + \ds(x)\frac{k-\cmis-1}{k} \\
& =    (k-\dc(x)+\ds(x))\frac{k-\cmis-1}{k}\\
& =    (k-\dd(x))\frac{k-\cmis-1}{k}\\
& \geq \dr(x)\frac{k-\cmis-1}{k}
\end{align*}
Hence, since no value is transferred from $x$ to $e'$ in Step 2(a), each child edge of $x$ belonging to $E_-$ receives a value of at least $\frac{k-\cmis-1}{k}$ from $x$ in Step 2(b).
In particular, 
\begin{equation*}
m_x(e) \geq \frac{k-\cmis-1}{k}
\end{equation*}

\paragraph{Case 2: $e'\in\es$ or $e'$ does not exist.}
In this case, since $e' \not\in \ed$, $x$ has at least
 $k-\cmis-\ds(x)$ high-colored child edges belonging to $\ed$.
By Fact~4, $x$ receives a value of at least $\frac{k-\dc(x)}{k}$ from
 the child edges of each of these edges.
Note that this value comes solely from child edges of $x$'s
high-colored child edges, not from the high-colored edges
 themselves. 
Therefore, by Fact~3, there is also a contribution of
 $\frac1k$ from each of $x$'s child edges belonging to $\ed$.
Finally, there are at least $\ds(x)-1$ child edges of $x$ belonging to $\es$ (if $e'$ exists, there are $\ds(x)-1$ such edges, and otherwise there are $\ds(x)$ such edges). By Fact~1, each of these edges transfers a value of 1 to $x$. Thus, 
\begin{alignat*}{2}
m(x)
& \geq \: && (k-\cmis-\ds(x))\frac{k-\dc(x)}{k} + (\ds(x)-1) + \frac{\dd(x)}{k}\\
& = && (k-\cmis-\ds(x))\frac{k-\dc(x)}{k} +
       (\ds(x)-1) \left( \frac{k-\dc(x)}{k} + \frac{\dc(x)}{k} \right)
       + \frac{\dd(x)}{k} \\ 
& = && (k-\cmis) \frac{k-\dc(x)}{k} - \ds(x) \frac{k-\dc(x)}{k}
       + \ds(x)\frac{k-\dc(x)}{k} - \frac{k-\dc(x)}{k} \\
&   && + (\ds(x)-1)\frac{\dc(x)}{k} + \frac{\dd(x)}{k} \\
& = && (k-\cmis-1)\frac{k-\dc(x)}{k} + (\ds(x)-1)\frac{\dc(x)}{k}
       + \frac{\dd(x)}{k} \\ 
& = && (k-\cmis-1)\frac{k-\dc(x)}{k} + (\ds(x)-1)\frac{\dc(x)-1}{k} 
       + \frac{\ds(x)-1+\dd(x)}{k} \\ 
& = && (k-\cmis-1)\frac{k-\dc(x)}{k} + (\ds(x)-1)\frac{\dc(x)-1}{k} 
       + \frac{\dc(x)-1}{k} \\
& = && (k-\cmis-1)\frac{k-\dc(x)}{k} + \ds(x)\frac{\dc(x)-1}{k}\\
& \geq && \dr(x)\frac{k-\cmis-1}{k}, \text{ as in Case 1}
\end{alignat*}
Hence, since no value is transferred from $x$ to $e'$ in Step 2(a), each
 child edge of $x$ belonging to $E_-$ receives a value of at least $\frac{k-\cmis-1}{k}$ from $x$ in Step 2(b).
Thus,
\begin{equation*}
m_x(e) \geq \frac{k-\cmis-1}{k}
\end{equation*}

\paragraph{Case 3: $e'\in\er$.}
The only difference to Case 2 is that $x$ has exactly $\ds(x)$ child
 edges belonging to $\es$.
Thus,
\begin{align*}
m(x)
& \geq (k-\cmis-\ds(x))\frac{k-\dc(x)}{k} + \ds(x) + \frac{\dd(x)}{k}\\
& \geq \dr(x)\frac{k-\cmis-1}{k} +1, 
       \text{ using the same calculations as in Case 2} 
\end{align*}
Hence, since the value transferred from $x$ to $e'$ is smaller than 1,
 each child edge of $x$ belonging to $E_-$ receives a value larger than
 $\frac{k-\cmis-1}{k}$ from $x$.
Thus, again,
\begin{equation*}
m_x(e) \geq \frac{k-\cmis-1}{k}
\end{equation*}
\mbox{}\qed\end{proof}

By Theorems~\ref{detupper} and \ref{fftree}, an algorithm for \mctree
 can only be better than \FF, if it is both randomized and unfair.
However, the next result shows that even such algorithms cannot do much better than \FF.

\begin{theorem}
If $\RAND$ is a (possibly randomized) algorithm for \mck and $k\geq 2$, then
\begin{displaymath}
C_{\RAND}^{\text{\sc Tree}}(k)\leq\frac{k}{k+1}.
\end{displaymath} 
\end{theorem}
\begin{proof}
The adversary first reveals the edges of a path $P = \langle e_1,\ldots
, e_m \rangle$, for some large $m\in\mathbb{N}$.
Let $v_1,\ldots, v_{m+1}$ be the vertices in the path such that
 $e_i=(v_i,v_{i+1})$, for $1\leq i\leq m$. 
If $\E[\RAND(P)]\leq\frac{k}{k+1}m$, the adversary reveals no more
edges.  
If $\E[\RAND (P)]>\frac{k}{k+1}m$, then for each $i$, $1\leq i\leq m+1$,
 the adversary reveals $k$ edges constituting a star, $S_i$, with
 center vertex $v_i$. 
Let $S$ be the set consisting of the edges of every star $S_i$ for $1\leq i\leq m+1$.

If the adversary only reveals the edges of the path $P$, then $\E[\RAND (P)]\leq\frac{k}{k+1}m$ and so $\E[\RAND(P)]\leq\frac{k}{k+1}\OPT (P)$. Indeed, \OPT can color all $m$ edges in $P$, since $k\geq 2$ and so $\OPT(P)=m$. Assume now that the adversary also reveals the stars. In this case, \OPT rejects all edges of the path and instead colors the $k$ edges of each star. Thus, $\OPT(P\cup S)=k(m+1)$. Note that each of the edges $e_i=(v_i, v_{i+1})$ is incident to the center vertices of both $S_i$ and $S_{i+1}$. This implies that $\E[\RAND(S)]\leq k(m+1)-2\E[\RAND(P)]$. Using the assumption $\E[\RAND(P)] >\frac{k}{k+1}m$, we get that
\begin{align*}
\E[\RAND(P\cup S)]&=\E[\RAND(P)]+\E[\RAND(S)]\\
&\leq \E[\RAND(P)] + k(m+1)-2\E[\RAND(P)]\\
&\leq k(m+1)-\frac{k}{k+1}m\\
&=\frac{k(km+k+1)}{k+1}\\
&=\frac{k}{k+1}k(m+1)+\frac{k}{k+1}\\
&=\frac{k}{k+1}\OPT (P\cup S)+\frac{k}{k+1}
\end{align*}
Since $m$ can be arbitrarily large, this shows that $\RAND$ cannot be better than $\frac{k}{k+1}$-competitive.
\qed\end{proof}

We now show that 
the competitive ratio of any fair algorithm tends to 1 as $k$ tends to infinity.

\begin{theorem}
\label{lowktree}
If $\FAIR$ is a fair algorithm, then for any $k \geq 2$,
\begin{displaymath}
C_{\FAIR}^{\text{{\sc Tree}}}(k)\geq\frac{2\sqrt{k}-2}{2\sqrt{k}-1}.
\end{displaymath}
\end{theorem}
\begin{proof}
Assume first that \FAIR is a deterministic algorithm.
Let $T=(V,E)$ be a tree and assume that the edges of $T$ have been revealed to \FAIR in some order. 
For the analysis, we will view $T$ as a rooted tree by choosing an arbitrary vertex to be the root. 
As in the proof of Theorem~\ref{fftree}, we let $e=(x,y)$ imply that $x$ is the parent of $y$.

We will apply the charging technique from Section \ref{lowertech} to show that $\FAIR$ is $C$-competitive, where $C=\frac{2\sqrt{k}-2}{2\sqrt{k}-1}$. 
We will use the notation introduced 
just before Theorem~\ref{fftree}.
Recall that all edges in $\ec$ have an initial value of $1$. Edges in $\ed$ have a surplus of $1-C$ and edges in $\es$ have a surplus of $1$. Edges in $\er$ have an initial value of $0$. The goal is to distribute the surplus from $\ec$ among the edges in $\er$ so that all of them get a final value of at least $C$. To this end, we use the following strategy: 

\begin{enumerate}[Step 1:]
\item  Each edge $(v,u)\in \ec$ transfers its surplus to its
  parent vertex, $v$.\\
 For each vertex $v$, let $m(v)$ denote the value transferred to $v$
  in this step. 
\item Consider in turn all vertices $v\in V$. 
  \begin{enumerate}
  \item If the vertex $v$ has a parent edge $e'\in\er$, then $v$ transfers a
    value of $\min\left\{m(v), C\right\}$ to $e'$.
  \item Any value remaining at $v$ is distributed equally among the
    child edges of $v$ belonging to $\er$.
  \end{enumerate}
  For each edge $e$, let $m_v(e)$ denote the value transferred from $v$ to $e$ in this step.
\end{enumerate}
This finishes the description of the strategy. 

Fix an edge $e=(x,y)\in \er$. 
In Step 1, $y$ receives $m(y) = \dc(y)-C\dd(y)$.
Thus, in Step 2(a), $e$ receives 
 $$m_y(e) = \min\{C,\dc(y)-C\dd(y)\}$$
from $y$.
We will show that $m_x(e)+m_y(e)\geq C$. 
If $m_y(e) \geq C$, this is clearly true.
Thus, we may assume that $\dc(y)-C\dd(y)<C$.
Note that
\begin{align*}
\dc(y)-C\dd(y)<C &\Rightarrow \dc(y)<C(\dd(y)+1)< \dd(y)+1\\
&\Rightarrow \dc(y)-\dd(y)<1\\
&\Rightarrow \dc(y)=\dd(y).
\end{align*}
It follows that we only need to consider the case where $\dc(y)=\dd(y)$, meaning that all of the edges incident to $y$ which have been colored by $\FAIR$ have also been colored by \OPT. This implies that the value transferred to $e$ from its colored child edges is  
 $$m_y(e) = (1-C)\dc(y)\,.$$ 
When calculating a lower bound on $m_x(e)$, we consider four cases.
In each case, we use the following two simple facts.

\begin{enumerate}[Fact 1:]
\item[Fact 1:] $\dd(x)+\dr(x) \leq k$.
\item[Fact 2:] $\dc(x)+\dc(y) \geq k$.
\end{enumerate}

{\em Proof of Fact 1:}
Note that $\dd(x)+\dr(x)$ is exactly the number of edges incident to $x$ that are colored by \OPT.
Thus, Fact 1 follows trivially, since no algorithm can color more than $k$ edges incident to $x$.

{\em Proof of Fact 2:}
This follows from the fact that the edge $(x,y)$ is rejected by the fair algorithm \FAIR.

In what follows, we will rely on the following elementary fact:
Consider a quadratic polynomial $ax^2+bx+c$ with $a,b,c \in \mathbb{R}$ and $a>0$.
If the discriminant $D = b^2-4ac = 0$, then the polynomial is non-negative.

\paragraph{Case 1: The parent edge of $x$ belongs to $\er$.} 
In this case,
\begin{align}
\notag m_x(e) & \geq \frac{m(x)-C}{\dr(x)-1} = \frac{\dc(x)-C\dd(x)-C}{\dr(x)-1}\\
\label{case1mx}       & \geq \frac{\dc(x)-C\dd(x)-C}{k-\dd(x)-1}, \text{ by Fact 1.}
\end{align}
Thus, we obtain the following, where the second inequality follows from Fact~2, and the third inequality comes from $\dd(x)\leq \dc(x)$:
\begin{align}
\notag m_x(e)&+m_y(e) \geq \frac{\dc(x)-C\dd(x)-C}{k-\dd(x)-1} + (1-C)\dc(y)\\
\notag & \geq \frac{\dc(x)-C\dd(x)-C}{k-\dd(x)-1} + (1-C)(k-\dc(x))
\\
\notag & = \frac{\dc(x)-C\dd(x)-C+(k-\dd(x)-1)(1-C)(k-\dc(x))}{k-\dd(x)-1}\\
\notag & \geq \frac{\dc(x)-C\dd(x)-C+(k-\dc(x)-1)(1-C)(k-\dc(x))}{k-\dd(x)-1}\\
\notag & = \frac{\dc(x)-C\dd(x)-C+(1-C)(k-\dc(x))^2+(C-1)(k-\dc(x))}{k-\dd(x)-1}\\
\notag&= \frac{(1-C)(k-\dc(x))^2+(C-2)(k-\dc(x))+(1-C)k}{k-\dd(x)-1}+C\\
&\geq C. \label{case1fair}
\end{align}
Here, the final inequality (\ref{case1fair}) holds since the numerator of the fraction is a quadratic polynomial in $(k-\dc(x))$ whose discriminant is zero:
\begin{align*}
(C-2)^2-4\cdot (1-C)\cdot(1-C)k&=\left(\frac{-2\sqrt{k}}{2\sqrt{k}-1}\right)^2-\frac{4k}{\left(2\sqrt{k}-1\right)^2}=0.
\end{align*} 

\paragraph{Case 2: The parent edge of $x$ belongs to $\es$.} 
In this case,
\begin{align*}
m_x(e) & = \frac{m(x)}{\dr(x)} = \frac{(\dc(x)-1)-C\dd(x)}{\dr(x)}\\
       & \geq \frac{\dc(x)-C\dd(x)-1}{k-\dd(x)}, \text{ by  Fact 1.}
\end{align*}
Thus, we obtain the following, where the second inequality follows from Fact~2 and the third inequality comes from $\dd(x) = \dc(x)-\ds(x) \leq \dc(x)-1$:
\begin{align} 
m_x(e) + m_y(e)
&\notag \geq \frac{\dc(x)-C\dd(x)-1}{k-\dd(x)} + (1-C)\dc(y)\\
&\notag \geq \frac{\dc(x)-C\dd(x)-1}{k-\dd(x)} + (1-C)(k-\dc(x))
%, \text { by Fact 2}
\\
&\notag = \frac{\dc(x)-C\dd(x)-1 + (1-C)(k-\dc(x))(k-\dd(x))}{k-\dd(x)}\\
&\notag \geq \frac{\dc(x)-C\dd(x)-1 + (1-C)(k-\dc(x))(k-\dc(x)+1)}{k-\dd(x)}\\
&\notag =\frac{(1-C)(k-\dc(x))^2-C(k-\dc(x))+k-1-C\dd(x)}{k-\dd(x)}\\
&\notag=\frac{(1-C)(k-\dc(x))^2-C(k-\dc(x))+(1-C)k-1}{k-\dd(x)}+C\\
&\geq C. \label{case2fair}
\end{align}
Here, the final inequality (\ref{case2fair}) holds since the numerator of the fraction is a quadratic polynomial in $(k-\dc(x))$ whose discriminant is zero: 
\begin{align*}
C^2-4(1-C)((1-C)k-1)&=C^2-4\frac{1}{2\sqrt{k}-1}\left(\frac{(\sqrt{k}-1)^2}{2\sqrt{k}-1}\right)\\
&=C^2-4\frac{(\sqrt{k}-1)^2}{(2\sqrt{k}-1)^2}=0.
\end{align*}

\paragraph{Case 3: The parent edge of $x$ belongs to $\ed$.} 
In this case, we have 
 \begin{align*}
 m_x(e) & = \frac{m(x)}{\dr(x)} = \frac{(\dc (x)-1)-C(\dd(x)-1)}{\dr(x)}\\
        & \geq \frac{\dc (x)-C\dd(x)+C-1}{k-\dd(x)}, \text{ by  Fact 1.}
\end{align*}
Recall that $m_y(e) = \dc(y)(1-C)$. Thus, if $\dc(y)(1-C)\geq C$, we are done. Hence, we assume from now on that $\dc(y)<\frac{C}{1-C}=2\sqrt{k}-2$. By Fact~2, this implies that $\dc(x)> k-(2\sqrt{k}-2)$.
Therefore, $(1-C)(k-\dc(x))< (1-C)(2\sqrt{k}-2)=C$ which implies the second to last inequality below. The second inequality below follows from Fact~2 and the third inequality comes from $\dd(x)\leq \dc(x)$.
\begin{align}
\notag m_x(e)+m_y(e)
 &\geq\frac{\dc (x)-C\dd(x)+C-1}{k-\dd(x)}+ (1-C)\dc(y)\\
\notag &\geq\frac{\dc (x)-C\dd(x)+C-1}{k-\dd(x)}+ (1-C)(k-\dc(x)), \text{ by Fact 2 }\\
\notag &= \frac{\dc (x)-C\dd(x)+C-1 + (1-C)(k-\dc(s))(k-\dd(x))}{k-\dd(x)}\\
\notag &\geq \frac{\dc (x)-C\dd(x)+C-1 + (1-C)(k-\dc(x))(k-\dc(x))}{k-\dd(x)}\\
\notag &> \frac{\dc(x)-C\dd(x)-1 + (1-C)(k-\dc(x))(k-\dc(x)+1)}{k-\dd(x)}\\
\label{case3final} &\geq C.
\end{align}
Here, 
the final inequality (\ref{case3final}) follows exactly as in Case 2.

\paragraph{Case 4: The parent edge of $x$ does not exist.} 
In this case,
 $$m_x(e) = \frac{m(x)}{\dr(x)} = \frac{\dc(x)-C\dd(x)}{\dr(x)} > \frac{\dc (x)-C\dd(x)+C-1}{k-\dd(x)}\,.$$
Thus, $m_x(e)+m_y(e)\geq C$ follows as in Case 3.

\paragraph{Randomized algorithms.}
Assume now that \FAIR is a randomized algorithm.
The above analysis holds for any coloring that \FAIR may produce.
Hence, for any coloring produced by \FAIR, the number of colored edges
 is at least $(2\sqrt{k}-2)/(2\sqrt{k}-1)$ times the number of edges
 colored by \OPT.
Clearly, this means that the expected number of edges colored by \FAIR
 is at least  $(2\sqrt{k}-2)/(2\sqrt{k}-1)$ times the number of edges
 colored by \OPT.
\qed\end{proof}

We will show that the lower bound of Theorem \ref{lowktree} is essentially tight by providing a matching upper bound on the competitive ratio of \NF when $k$ is a square number. To this end, we will use the following result
 from \cite{kedge}.  
\begin{lemma}[Favrholdt and Nielsen~\cite{kedge}]
If the edges of a graph are colored in such a way that each color is used exactly $n$ or $n+1$ times for some $n\in\mathbb{N}$, then there exists an ordering of the edges such that \NF produces an equivalent coloring.
\label{lmNFlemma}
\end{lemma}

The following corollary follows easily from Lemma \ref{lmNFlemma}.

\begin{corollary}
\label{cor:next-fit}
Consider a graph, $G=(V,E)$, and a coloring, \coloring, of all edges of $G$ using
 at most $k$ colors.
Let $H$ be a graph consisting of $k$ disjoint copies of $G$. 
There exists an ordering of the edges of $H$ such that, for each of
 the $k$ copies of $G$ in $H$, the coloring produced by \NF is
 equivalent to \coloring.
\label{corolNF}
\end{corollary}

\begin{proof}
Let $G_1$, $G_2$, \ldots, $G_k$ denote the $k$ copies of $G$.
Furthermore, let $\coloring_1$, $\coloring_2$, \ldots, $\coloring_k$
 be the $k$ colorings that can be obtained from \coloring by cyclic
 permutations of the colors $1,2,\ldots,k$.
If, for $1 \leq i \leq k$, $G_i$ is assigned the coloring
 $\coloring_i$, we obtain a coloring of $H$ where all colors are used
 the same number of times.
The result now follows from Lemma~\ref{lmNFlemma}.
\qed\end{proof}

Note that Corollary \ref{corolNF} implies that if $\mathcal{G}$ is
some family of graphs and $\mathcal{G}$ is closed under disjoint
union, then \NF has the worst possible competitive ratio among fair
algorithms for \mcarg{$\mathcal{G}$}. 
This can be seen in the following way:
For any graph, $G$, and any coloring, \coloring, of $G$ produced by a
 fair algorithm, the adversary can do the following:
\begin{itemize}
\item Make $k$ copies of $G$, resulting in a graph $H$.
\item Give the edges of $H$ corresponding to the colored edges of \coloring.
 According to Corollary~\ref{cor:next-fit}, these edges can be given in order, such that the edges of each copy of $G$ receives a coloring equivalent to \coloring.
\item Give the edges of $H$ corresponding to edges that were not colored by \coloring.
 Since \coloring was produced by a fair algorithm, \NF will not be able to color any of these edges.
\end{itemize}
Hence, for any sequence, $E_G$, of edges and any fair algorithm \FAIR, there is a sequence, $E_H$, of edges, such that \NF uses just as many colors on $E_H$ as \FAIR does on $E_G$, and the optimal number of colors is the same for both sequences.

Even though \text{\sc Tree} is not closed under disjoint union, a forest consisting of $k$ trees may be made into a single tree by revealing $k-1$ edges connecting the $k$ trees. 
Since this will add at most $k-1$ to the number of edges colored
by \NF, we may still apply Corollary \ref{corolNF} for the class {\sc Tree}.

\begin{theorem}
\label{upperNFtrees}
For $k\geq 4$, 
\begin{equation*}
C_{\NFm}^{\text{{\sc Tree}}}(k)\leq\frac{\frac{k}{\lceil\sqrt{k}\rceil}+\lceil\sqrt{k}\rceil-2}{\frac{k}{\lceil\sqrt{k}\rceil}+\lceil\sqrt{k}\rceil-1}.
\end{equation*}
In particular, if $k=n^2$ for some integer $n\geq 2$, then \NF is a worst possible
fair algorithm with
\begin{equation*}
  C_{\NFm}^{\text{{\sc Tree}}}(k) = \frac{2\sqrt{k}-2}{2\sqrt{k}-1}.
\label{upperNF}
\end{equation*}
\end{theorem}
\begin{proof}
The lower bound for the case where $k$ is a square number follows from Theorem~\ref{lowktree}. 
For the upper bound, we define a tree $T=(V,E)$ and a subset
 $E' \subset E$.
We specify a coloring, \coloring, of $E'$ with the property that
 each edge in $E \setminus E'$ is adjacent to edges of all $k$ colors.

We first describe $E'$ and \coloring.
The tree $T$ contains $N$ \emph{bunches} of stars, for some large $N$. 
Each bunch consists of a set of stars:
\begin{itemize}
\item One {\em large} star with $k-{\ceilk}$ edges
  colored with $\colorset_{1,k-\ceilk}$.\\
 The center vertex of the large star in bunch $i$, $1 \leq i \leq N$, is
  called $v_i$.
\item ${\ceilk}-1$ {\em small} stars, each with $\ceilk$ edges
 colored with $\colorset_{k-\ceilk+1,k}$\,.
\end{itemize}

We now describe $E \setminus E'$.
For each $i$, $1 \leq i \leq N$, $E \setminus E'$ contains an edge between $v_i$ and the center vertex of each of the small stars in bunch $i$.
For $1 \leq i < N$, there is an edge from $v_{i+1}$ to the center vertex of one of the 
 small stars in the $i$th bunch. 
Note that, after assigning the coloring \coloring to $E'$, none of the
 edges in $E \setminus E'$ can be colored.

The adversary will use $k$ disjoint copies, $T_1=(V_1,E_1),\ldots , T_k=(V_k,E_k)$, of
 $T$. 
For each $T_i$, let $E'_i$ denote the set of edges corresponding to
 $E'$ and let $T'_i=(V_i,E'_i)$.
If the edges of $E'_i \cup E'_2 \cup \ldots \cup E'_k$ are given
 first, it follows from Corollary~\ref{corolNF} that they can be given
 an order such that the coloring produced by \NF on each $T'_i$ is
 equivalent to \coloring.
Afterwards, no other edges can be colored.

Finally, the $k$ disjoint trees are connected, using $k-1$ edges
 between vertices that have degree one in the trees.
The resulting tree is called $\mathcal{T}$.

Since $k\geq 4$, we must have $\ceilk +2\leq k$ and so the maximum
degree of the graph is $k$. 
Thus, since the graph has no cycles, \OPT colors all edges of the graph.

\NF colors
\begin{align*}
\nf(\mathcal{T}) &= kN\big(k-\ceilk+(\ceilk-1)\ceilk\big)+k-1\\
& = kN\left(k+\ceilk^2-2\ceilk \right)+k-1
\end{align*}
 edges and rejects $k\big(N(\ceilk-1)+N-1\big) = kN\ceilk-k$ edges. 
Since \OPT colors all edges in the graph, 
 $$\OPT(\mathcal{T})=kN\left(k+\ceilk^2-\ceilk \right)-1\,.$$ 
Thus, 
\begin{align*}
\nf(\mathcal{T})&\leq\frac{k+\ceilk^2 -2\ceilk}{k+\ceilk^2-\ceilk}\OPT(\mathcal{T})+k\\
&= \frac{\frac{k}{\lceil\sqrt{k}\rceil}+\lceil\sqrt{k}\rceil-2}{\frac{k}{\lceil\sqrt{k}\rceil}+\lceil\sqrt{k}\rceil-1}\OPT(\mathcal{T})+k
\end{align*} 
Since $N$ can be arbitrarily large, the result follows.
\qed\end{proof}

Theorem~\ref{upperNFtrees} shows that the bound of Theorem~\ref{lowktree} is tight whenever $k$ is a square number. We will briefly consider the case where $k$ is not a square
number. Any fair algorithm for {\sc Edge-$1$-Coloring(Tree)} is just
the greedy matching algorithm. It is observed in several papers that this algorithm is $\frac12$-competitive (for all input graphs) and that no deterministic algorithm can do better, even when the input graph is a tree.  If $k\geq 2$, but not a square number, then the lower bound from Theorem \ref{lowktree} can be slightly improved by using the fact that $\dc(x)$ must be an integer. In particular, for $k=2$, it follows from Theorem~3.1 in~\cite{kedge} that any fair algorithm is $\frac12$-competitive on any class of graphs. Combining this result with Theorem~\ref{detupper} shows that on trees, the competitive ratio of any fair algorithm is exactly $\frac12$.
We show in Theorems~\ref{FAIRk3} and \ref{NFk3} that any fair algorithm for {\sc Edge-$3$-Coloring(Tree)} is $\frac58$-competitive, and the competitive ratio of \NF is exactly $\frac58$.
Thus, for $k \leq 4$, we have completely tight bounds.
For $k \geq 5$, the difference between our upper and lower bounds is less than $0.0153$ and tends to $0$ as $k$ tends to infinity. Also, we get that \FF has a strictly better competitive ratio than \NF on trees whenever $k\geq 3$.

\begin{theorem}
\label{FAIRk3}
If $\FAIR$ is a fair algorithm, then $C_{\FAIR}^{\text{{\sc Tree}}}(3)\geq \frac{5}{8}$.
\end{theorem}
\begin{proof}
In order to prove that all fair algorithms are $\frac{5}{8}$-competitive on trees when $k=3$, we modify Step~1 of the strategy used in Theorem~\ref{lowktree} for distributing the surplus.  Step~2 is unmodified, but for convenience, we give both steps. Let $C=\frac{5}{8}$.
\begin{enumerate}[Step 1:]
\item Each edge $(v,u)\in E_+^s$ transfers a value of $\frac{7}{8}$ to its parent vertex, $v$, and a value of $\frac{1}{8}$ to its child vertex, $u$. Each edge $(v,u)\in E_+^d$ transfers its surplus of $1-C=\frac{3}{8}$ to its parent vertex, $v$.\\
 For each vertex $v$, let $m(v)$ denote the value transferred to $v$
  in this step. 
\item Consider in turn all vertices $v\in V$. 
  \begin{enumerate}
  \item If the vertex $v$ has a parent edge $e'\in\er$, then $v$ transfers a
    value of $\min\left\{m(v), C\right\}$ to $e'$.
  \item Any value remaining at $v$ is distributed equally among the
    child edges of $v$ belonging to $\er$.
  \end{enumerate}
  For each edge $e$, let $m_v(e)$ denote the value transferred from $v$ to $e$ in this step.
\end{enumerate}
This finishes the description of the strategy. 

Fix an edge $e=(x,y)\in E_{-}$. We need to show that $m_x(e)+m_y(e)\geq C=\frac{5}{8}$. First, note that if $d_+^s(y)\geq 1$, then $m(y)\geq \frac{7}{8}$ and we are done. Also, if $d_+^d(y)\geq 2$, then $m(y)\geq 2(1-C)=\frac{6}{8}$, and again we are done. Thus, we may assume that $d_+^s(y)=0$ and $d_+^d(y)\leq 1$. We now show that $e$ receives a value of at least $\frac{5}{8}$ in all such cases.

\paragraph{Case 1: $d_+(y)=1$.} In this case, $m_y(e)=\frac{3}{8}$, so we just need to show that $m_x(e)\geq \frac{2}{8}$.
Note that $d_+(x) \geq k - d_+(y) = 2$. Thus, $x$ has at least one child edge belonging to $E_+$.
\begin{itemize} %d_+(y)=1
\item {\em Case 1.1: The parent edge of $x$ belongs to $E_+^s$.}
  In this case, the parent edge of $x$ contributes a value of $\frac18$ to $m(x)$. 
\begin{itemize}%d_+(y)=1, parent\in E_+^s
\item {\em Case 1.1.1: $d_{-}(x)=3$.} In this case, $d^s_{+}(x)=d_{+}(x)\geq 2$, and therefore at least one child edge of $x$ belongs to $E_+^s$. It follows that $m_x(e)\geq \frac{1}{3} (\frac{1}{8}+\frac{7}{8})=\frac{1}{3} > \frac{2}{8}$.
\item {\em Case 1.1.2: $d_{-}(x)\leq 2$.} Since $x$ has at least one child edge in $E_+$, $m_x(e)\geq\frac{1}{2}(\frac{1}{8}+\frac{3}{8})=\frac{2}{8}$.
\end{itemize}%d_+(y)=1, parent\in E_+^s

\item {\em Case 1.2: The parent edge of $x$ belongs to $E_+^d$.}
  Since in this case, $d_+^d(x) \geq 1$, it follows that $d_-(x) \leq 2$.
  Thus, we have only the following two subcases:
\begin{itemize}%d_+(y)=1, parent\in E_+^d
\item {\em Case 1.2.1: $d_{-}(x)=2$.} At least one child edge of $x$ must belong to $E_+^s$. Thus, $m_x(e)\geq \frac{1}{2}\cdot\frac{7}{8}>\frac{3}{8}$.
\item {\em Case 1.2.2: $d_{-}(x)=1$.} Since at least one child edge of $x$ belongs to $E_+$, $m_x(e)\geq \frac{3}{8}$.
\end{itemize}%d_+(y)=1, parent\in E_+^d

\item {\em Case 1.3: The parent edge of $x$ belongs to $E_-$.}
  In this case, $x$ has at most two child edges belonging two $E_-$.
  Furthermore, $x$ has at least two child edges in $E_+$ and at least one of them belongs to $E_+^s$.
  Thus, $m_x(e) \geq \frac12 (\frac38 + \frac78 - \frac58) > \frac28$.

\item {\em Case 1.4: $x$ has no parent edge.} Since $d_{+}(x)\geq 2$, $m_x(e)\geq \frac13 \cdot 2\cdot \frac{3}{8} = \frac28$.
\end{itemize}%d_+(y)=0

\paragraph{Case 2: $d_{+}(y)=0$.} In this case, $d_+(x)=3$, so $x$ has at least two child edges belonging to $E_+$. 
Furthermore, $d_+^s(x) \geq 1$.
We show that $m_x(e) \geq \frac58$ in all subcases.
\begin{itemize}
\item {\em Case 2.1: The parent edge of $x$ belongs to $E_+^s$.}
\begin{itemize}
\item {\em Case 2.1.1: $d_-(x)=3$.} In this case, $d_+^s(x)=3$, and hence $m_x(e)=\frac{1}{3}(\frac{1}{8}+2\cdot \frac{7}{8})=\frac{5}{8}$.
\item {\em Case 2.1.2: $d_-(x)= 2$.} In this case, $d_+^s(x)\geq2$, and hence, $m_x(e)\geq \frac{1}{2}(\frac{1}{8}+\frac{7}{8}+\frac{3}{8})> \frac{5}{8}$.
\item {\em Case 2.1.2: $d_-(x)= 1$.} Two child edges of $x$ must belong to $E_+$ and so $m_x(e)\geq \frac{1}{8}+2\cdot \frac{3}{8}> \frac{5}{8}$.
\end{itemize}%d_+(y)=0, parent\in E_+^s

\item {\em Case 2.2: The parent edge of $x$ belongs to $E_+^d$.}
In this case, $d_-(x) \leq 2$. Thus, $m_x(e) \geq \frac12(\frac78 + \frac38) = \frac58$.

\item {\em Case 2.3: The parent edge of $x$ belongs to $E_-$.} In this case, $x$ has at least three child edges in $E_+$, and at least two of them belong to $E_+^s$. Moreover, $x$ has at most two child edges belonging to $E_-$. Thus, $m_x(e)\geq \frac{1}{2}(2\cdot\frac{7}{8}+ \frac38-\frac{5}{8})=\frac{6}{8}$.
\item {\em Case 2.4: $x$ has no parent edge.}
  \begin{itemize}
  \item {\em Case 2.4.1: $d_-(x) = 3$.}
    In this case, $d_+^s(x) = 3$.
    Thus, $m_x(e) = \frac13 \cdot 3 \cdot \frac78 = \frac78$.
  \item {\em Case 2.4.2: $d_-(x) \leq 2$.} 
Since $d_+^s(x) \geq 1$, $m_x(e) \geq \frac{1}{2}(\frac78 + 2 \cdot \frac38) > \frac68$.
  \end{itemize}
\end{itemize}
\qed\end{proof}

We now show that the analysis in Theorem~\ref{FAIRk3} is tight by showing that $C_{\NFm}^{\text{{\sc Tree}}}(3)=\frac{5}{8}$. This is done by creating an adversary graph which combines the two cases (cases 1.1.2 and 2.1.1) from the proof of Theorem~\ref{FAIRk3} for which the strategy used for distributing the surplus could only guarantee a value of exactly $\frac{5}{8}$.

\begin{theorem}
\label{NFk3}
 $C_{\NFm}^{\text{{\sc Tree}}}(3)=\frac{5}{8}$.
\end{theorem}
\begin{proof}
The lower bound follows from Theorem~\ref{FAIRk3}.
For the upper bound, let $N$ be an integer divisible by $3$. 
The adversary graph for $N=3$ is illustrated in Fig.~\ref{fig:nfk3}.
The adversary first reveals $2N$ isolated edges (shown as the top vertical edges in Fig.~\ref{fig:nfk3}). 
For $1 \leq i \leq 3$, denote by $M_i$ the subset of these $2N$ edges colored with the color $i$ by \NF. Then, the adversary reveals a path $P=\langle e_1,\ldots , e_{N+1}\rangle$ consisting of $N+1$ new edges, revealing the edges from left to right. An \emph{inner vertex} of $P$ is a vertex of degree $2$. For each inner vertex $v$ of $P$, the adversary reveals four edges $(v,u_1),\ldots , (v,u_4)$. Note that \NF colors $(v,u_1)$ with the unique color $c\in \mathcal{C}_{1,3}\setminus \mathcal{C}_v$ at $v$ and rejects the other three edges. The adversary then reveals an edge $(u_1, u_1')$ which is colored with the color $(c+1) \bmod 3$ by \NF. Finally, the adversary picks two distinct isolated edges $(x,y),(x',y')\in M_{(c+2) \bmod 3}$ and reveals two new edges $(u_1, y)$ and $(u_1, y')$. \NF rejects both of these edges. The adversary continues with the next inner vertex (unless $v$ was the last inner vertex) and repeats the above procedure. Note that for each $c\in \mathcal{C}_{1,3}$, there are $N/3$ inner vertices on the path for which $\mathcal{C}_{1,3}\setminus \mathcal{C}_v=\{c\}$. Thus, the adversary does not run out of edges to pick from $M_c$ (in fact, it uses all $\frac{2}{3}N$ of them). This finishes the description of the adversary strategy. 

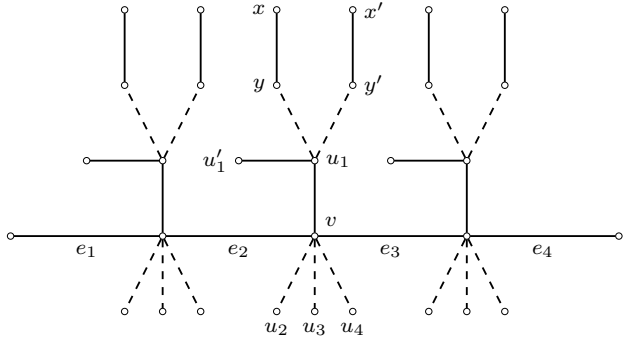
\begin{figure}
\centering
\begin{tikzpicture}
	\begin{pgfonlayer}{nodelayer}
		\node [style=rn] (0) at (-2, -0) {};
		\node [style=rn] (1) at (-4, -0) {};
		\node [style=rn] (2) at (-2, 1) {};
		\node [style=rn] (3) at (-3, 1) {};
		\node [style=rn] (4) at (-2.5, 2) {};
		\node [style=rn] (5) at (-1.5, 2) {};
		\node [style=rn] (6) at (-2.5, 3) {};
		\node [style=rn] (7) at (-1.5, 3) {};
		\node [style=rn] (8) at (-2, -1) {};
		\node [style=rn] (9) at (-2.5, -1) {};
		\node [style=rn] (10) at (-1.5, -1) {};
		\node [style=rn, label={above right:$v$}] (11) at (0, -0) {};
		\node [style=rn, label={below:$u_3$}] (12) at (0, -1) {};
		\node [style=rn, label={below:$u_2$}] (13) at (-0.5, -1) {};
		\node [style=rn, label={left:$u'_1$}] (14) at (-1, 1) {};
		\node [style=rn, label={below:$u_4$}] (15) at (0.5, -1) {};
		\node [style=rn, label={right:$u_1$}] (16) at (0, 1) {};
		\node [style=rn, label={right:$y'$}] (17) at (0.5, 2) {};
		\node [style=rn, label={left:$y$}] (18) at (-0.5, 2) {};
		\node [style=rn, label={left:$x$}] (19) at (-0.5, 3) {};
		\node [style=rn, label={right:$x'$}] (20) at (0.5, 3) {};
		\node [style=rn] (21) at (2.5, -1) {};
		\node [style=rn] (22) at (1.5, 2) {};
		\node [style=rn] (23) at (2.5, 3) {};
		\node [style=rn] (24) at (2, 0) {};
		\node [style=rn] (25) at (2, -1) {};
		\node [style=rn] (26) at (1, 1) {};
		\node [style=rn] (27) at (2, 1) {};
		\node [style=rn] (28) at (1.5, -1) {};
		\node [style=rn] (29) at (2.5, 2) {};
		\node [style=rn] (30) at (1.5, 3) {};
		\node [style=rn] (31) at (4, -0) {};
	\end{pgfonlayer}
	\begin{pgfonlayer}{edgelayer}
		\draw [style=ed] (0) to (2);
		\draw [style=ed] (2) to (3);
		\draw [style=ed] (1) to node[label={below:$e_1$}]{} (0);
		\draw [style=ed] (0) to node[label={below:$e_2$}]{} (11);
		\draw [style=ed] (11) to node[label={below:$e_3$}]{} (24);
		\draw [style=ed] (24) to node[label={below:$e_4$}]{} (31);
		\draw [style=ed] (6) to (4);
		\draw [style=ed] (5) to (7);
		\draw [style=ed] (19) to (18);
		\draw [style=ed] (17) to (20);
		\draw [style=ed] (30) to (22);
		\draw [style=ed] (29) to (23);
		\draw [style=ned] (2) to (4);
		\draw [style=ned] (2) to (5);
		\draw [style=ned] (0) to (9);
		\draw [style=ned] (0) to (8);
		\draw [style=ned] (0) to (10);
		\draw [style=ned] (11) to (13);
		\draw [style=ned] (11) to (12);
		\draw [style=ned] (11) to (15);
		\draw [style=ned] (24) to (28);
		\draw [style=ned] (24) to (25);
		\draw [style=ned] (24) to (21);
		\draw [style=ned] (27) to (22);
		\draw [style=ned] (27) to (29);
		\draw [style=ned] (16) to (18);
		\draw [style=ned] (16) to (17);
		\draw [style=ed] (24) to (27);
		\draw [style=ed] (27) to (26);
		\draw [style=ed] (11) to (16);
		\draw [style=ed] (16) to (14);
	\end{pgfonlayer}
\end{tikzpicture}
\caption{The adversary graph used in the proof of Theorem~\ref{NFk3} when $N=3$. Solid edges are colored by \NF and dashed edges are rejected by \NF.}
\label{fig:nfk3}
\end{figure}

\NF colors the $2N$ isolated edges, the $N+1$ edges of the path $P$, and for each inner vertex $v$ it colors $(v,u_1)$ and $(u_1,u_1')$. Thus, $\nf(I)=2N+(N+1)+2N=5N+1$. On the other hand, $\OPT$ 
rejects all edges of the path $P$. Furthermore, for each inner vertex $v$, $\OPT$ rejects $(v, u_1)$.
The remaining edges form a graph with maximum degree $3$ and hence, \OPT can color all of these $8N$ edges.
It follows that $\nf(I)=\frac{5}{8}\OPT(I)+1$. This shows that $\NF$ cannot be better than $\frac{5}{8}$-competitive.
\qed\end{proof}

\section{Open Problems}
Finding optimal online algorithms for \mck in general and on other classes of graphs is an interesting open problem. We believe that the techniques used in the proofs of Theorems~\ref{fftree} and \ref{lowktree} can be generalized to, e.g., graphs of bounded degeneracy. In particular, graphs of bounded degeneracy can be oriented so that each vertex has bounded outdegree and the resulting digraph is acyclic. This makes it possible to use strategies for redistributing the surplus similar to the ones we have used for trees.

Deciding whether there is an algorithm better than \FF on trees would
also be interesting.
Such an algorithm could only be significantly better for small values
of $k$, and it would have to be both randomized and unfair.

\paragraph*{Acknowledgment.}
~The authors would like to thank the anonymous reviewers for helpful comments on this work and its presentation.

\bibliographystyle{plain}
\bibliography{refs}

\end{document}